\documentclass[journal,draftcls,onecolumn,12pt,twoside,letterpaper]{IEEEtran}
\usepackage{setspace}
\normalsize

\usepackage[top=1in, bottom=1in, left=1in, right=1in]{geometry}
\usepackage[utf8]{inputenc} 
\usepackage[T1]{fontenc}

\ifCLASSINFOpdf

\else

\fi
\usepackage[cmex10]{amsmath}
\usepackage{ifthen}
\usepackage{tikz}
\usepackage{cite}
\usepackage{amsthm,amsfonts,amssymb}
\usepackage{graphicx}
\usepackage{subfig}
\usepackage{blkarray}
\usepackage{dsfont}
\usepackage[mathscr]{euscript}
\usepackage{enumitem}
\usepackage{algorithm}
\usepackage[noend]{algpseudocode}
\usepackage{blkarray}
\usepackage{stfloats}%
\newtheorem{theorem}{Theorem}[section]
\newtheorem{corollary}{Corollary}[section]
\newtheorem{proposition}{Proposition}[section]
\newtheorem{lemma}{Lemma}[section]
\theoremstyle{remark}
\newtheorem*{remark}{Remark}
\theoremstyle{definition}

\newtheorem{example}{Example}[section]
\usepackage{url}
\usepackage{tcolorbox}
\usepackage{arydshln}
\usepackage{mathtools}
\usepackage{dsfont}

\definecolor{brickred}{cmyk}{0,0.89,0.94,0.28}
\definecolor{goldenrod}{cmyk}{0,0.10,0.84,0}
\definecolor{purple}{cmyk}{0.45,0.86,0,0}
\definecolor{rawsienna}{cmyk}{0,0.72,1,0.45}
\definecolor{olivegreen}{cmyk}{0.64,0,0.95,0.40}
\definecolor{peach}{cmyk}{0,0.5,0.7,0}
\definecolor{darkolive}{rgb}{0.,0.4,0.}
\colorlet{grey}{gray!40}

\global\long\def\E{\mathbb{E}}

\usepackage{graphicx}
\makeatletter
\newcommand{\ostar}{\mathbin{\mathpalette\make@circled\star}}
\newcommand{\make@circled}[2]{%
	\ooalign{$\m@th#1\smallbigcirc{#1}$\cr\hidewidth$\m@th#1#2$\hidewidth\cr}%
}
\newcommand{\smallbigcirc}[1]{%
	\vcenter{\hbox{\scalebox{0.77778}{$\m@th#1\bigcirc$}}}%
}
\makeatother

\usepackage[cmintegrals]{newtxmath}

\begin{document}

\title{On Achievable Rates Over Noisy Nanopore Channels}
\author{
\IEEEauthorblockN{V.~Arvind~Rameshwar,~\IEEEmembership{Member,~IEEE}\ and\ Nir Weinberger,~\IEEEmembership{Senior Member,~IEEE}}
\thanks{V. A. Rameshwar is with the India Urban Data Exchange Program Unit, Indian Institute of Science, Bengaluru, India, email: \texttt{arvind.rameshwar@gmail.com}. N. Weinberger is with the Department of Electrical and Computer Engineering, Technion, Haifa 3200003, Israel, email: \texttt{nirwein@technion.ac.il}. The research of N. Weinberger was partially supported by the Israel Science Foundation (ISF), grant no. 1782/22.}
}
\IEEEoverridecommandlockouts
%

\markboth{}%
{Rameshwar and Weinberger: Achievable Rates Over Noisy Nanopore Channels}

\maketitle

\begin{abstract}
	In this paper, we consider a recent channel model of a nanopore sequencer proposed by McBain, Viterbo, and Saunderson (2024), termed the \emph{noisy nanopore channel} (NNC). In essence, an NNC is a duplication channel with structured, Markov inputs, that is corrupted by memoryless noise. We first discuss a (tight) lower bound on the capacity of the NNC in the absence of random noise. Next, we present lower and upper bounds on the channel capacity of general noisy nanopore channels. We then consider two interesting regimes of operation of an NNC: first, where the memory of the input process is large and the random noise introduces erasures, and second, where the rate of measurements of the electric current (also called the sampling rate) is high. For these regimes, we show that it is possible to achieve information rates close to the noise-free capacity, using  low-complexity encoding and decoding schemes. In particular, our decoder for the regime of high sampling rates makes use of a change-point detection procedure -- a subroutine of immediate relevance for practitioners.
\end{abstract}



\IEEEpeerreviewmaketitle
\section{Introduction}
In the last decade, significant progress has been made in the problem of storing information on synthetically generated DNA strands \cite{dnachurch,dnagoldman,dnagrass,dnaerlich,dnayazdi,dnaorganick}, leading to widespread interest in DNA as a viable medium for the storage of archival data. In this light, various works, for example \cite{shomorony,nir_merhav,lenz} considered the fundamental information-theoretic limits of a channel model for DNA-based storage, which takes into account processes such as Polymerase Chain Reaction (PCR) amplification, random sampling from a pool of DNA strands, and subsequent reconstruction from noisy reads. Such a model assumes a sequencer that can only read short DNA strands, which are typically a few hundred bases long. More recently, the field of DNA sequencing witnessed a new revolution via nanopore sequencers \cite{sequencingsurvey,oxford} that can sequence DNA strands of lengths that are roughly $10$--$100$ Kilo-bases. We also refer the reader to other interesting experimental works on nanopore sequencers \cite{yazdi-nanopore,chakra-nanopore}.

Given the growing interest in nanopore sequencing, various papers\cite{mao,chandaknanopore,hamoumnanopore,mcbaininfo1,mcbainsurvey} proposed channel models for the sequencer, in an attempt to model the several sources of inaccuracies during reading. These include intersymbol interference (ISI), random dwell times of bases in the motor protein of the nanopore, ``backtracking'' and ``skipping'' (or equivalently, base insertions and deletions), fading, and so on. With the aid of simulation studies conducted using the Scrappie technology demonstrator \cite{scrappie} (now archived) of Oxford Nanopore Technologies, \cite{mcbaininfo1,mcbainsurvey} introduced a channel model that seemingly accurately models the physical nanopore channel at the raw signal (or sample) level. Essentially, such a \emph{noisy nanopore channel} (NNC) is given by the cascade of a duplication channel with a memoryless channel; further, the input to the duplication channel is a sequence of $\tau$-tuples of bases (also called $\tau$-mers), which has a specific Markov structure. The lumping of bases into $\tau$-mers models ISI, with $\tau$ representing the ``memory" (or ``stationarity") of the pore model; the duplication channel reflects the random dwell times of $\tau$-mers; and the memoryless channel is a model for the noise in the sequencing process. 

After the introduction of the NNC in \cite{mcbainmodel}, the authors in \cite{mcbaininfo2} established that the classical Shannon capacity, given by the maximum mutual information between (constrained) inputs and outputs, equals the channel capacity of the NNC (and, more generally, of noisy duplication channels with a Markov source). Furthermore, preliminary numerical estimates of the capacity were obtained for simple Markov-constrained noisy duplication channels; however, these do not accurately model the ISI effects in the NNC setting. This leaves open the question of accurately characterizing, or obtaining explicit estimates of, the capacity of the NNC. The goal of the current paper is to make some progress towards this goal.

In particular, we make use of tools and results from the literature on capacity computation for channels with synchronization errors (such as insertion and deletion channels) to obtain estimates of the capacity of selected NNCs. Starting from the seminal paper by Dobrushin \cite{Dob67}, much work has been carried out on such channels; a selection of papers on capacity computation over such channels is \cite{diggavidel,diggaviupper,drinea1,drinea3,drinea2,kirsch,fertonani,aravind,mercier,rahmati1}. On a related note, several works \cite{coding-1,coding-2,coding-3,coding-4} (see also \cite{sloanedel}) have also considered the question of constructing explicit codes over channels with synchronization errors.

In this paper, we first present a lower bound on the capacity of the NNC when there is no (memoryless) noise in the sequencing process. The proof of our bound for the \emph{noiseless} NNC is a much simplified presentation of a result that can also be adapted from the main result in \cite{kirsch}; the arguments in \cite{kirsch} in fact show that this lower bound is tight. Next, we present simple, computable lower and upper bounds on the capacity of general, noisy nanopore channels, which are the first, non-trivial bounds on the capacity of such channels. 

The bounds above are most effective for short memory lengths, and refining them for the long memory lengths required for practical channel models \cite{mcbainsurvey} appears to be challenging. We ameliorate this issue for the class of NNCs where the sequencing errors are introduced as erasures. For such a channel, we show that in fact information rates \emph{close to the noise-free capacity} are achievable, in the limit of large $\tau$-mer lengths, and in particular can be achieved by practical encoding and decoding algorithms. We then shift our attention to another regime of operation of general NNCs, wherein the outputs are read by sampling at a very high rate. The sampling rates are under the control of the system designer, who can set them to be as high as required, possibly at high cost \cite{mcbainsurvey}. Indeed, high sampling rates were also assumed in the early work \cite{mao} on nanopore channel modelling. For such channels, we show that information rates close to the noise-free capacity are achievable, again via simple encoding and decoding algorithms. Our decoding algorithm for the setting of high sampling rates uses a \emph{change-point detection} procedure for estimating the boundaries of runs of output symbols that arise from the same input. As this procedure is similar to that used in practice \cite{laszlo}, we believe that our analyses for these regimes will be of immediate interest for practitioners and theorists alike.
\section{Notation and Preliminaries}
\label{sec:notation}
\subsection{Notation}
For a positive integer $n$, we use $[n]$ as shorthand for $[1:n]$. Random variables are denoted by capital letters, e.g., $X, Y$, and small letters, e.g., $x, y$, denote their realizations. Sets are denoted by calligraphic letters, e.g., $\mathcal{X}, \mathcal{Y}$; the notation $\mathcal{X}^c$ denotes the complement of the set $\mathcal{X}$, when the universal set is clear from the context. Notation such as $P(x), P(y|x)$ are used to denote the probabilities $P_X(x), P_{Y|X}(y|x)$, when it is clear which random variables are being referred to. The notations ${H}(X) :=  \mathbb{E}[-\log P(X)], {H}(Y\mid X):= \mathbb{E}[-\log P(Y\mid X)]$, and $I(X;Y) := H(Y)-H(Y\mid X)$ denote the entropy of $X$, conditional entropy of $Y$ given $X$, and mutual information between $X$ and $Y$, respectively. Given any real $p\in [0,1]$, we let $h_b(p):=-p\log p -(1-p)\log(1-p)$, where $h_b$ denotes the binary entropy function; here, the base of the logarithm will be made clear from the context (we use $\ln$ to refer to the natural logarithm). The notation Ber$(p)$ and Bin$(n,p)$ refer, respectively, to the Bernoulli distribution with parameter $p$ and the Binomial distribution with parameters $n$ and $p$, where $p\in [0,1]$ and $n$ is a positive integer. Given sequences $(a_n)_{n\geq 1}$ and $(b_n)_{n\geq 1}$, we say that $a_n = O(b_n)$, if $a_n\leq C\cdot b_n$, for some fixed constant $C\geq 0$, for sufficiently large $n$, and $a_n = o(b_n)$, if $\lim_{n\to \infty} \frac{a_n}{b_n} = 0$.

Given a vector $\mathbf{b}\in \mathcal{X}^n$, for some finite alphabet $\mathcal{X}$ and integer $n\geq 1$, we let $\ell(\mathbf{b})$ denote its length. We define a \emph{run} of a symbol $x\in \mathcal{X}$ in $\mathbf{b}$ to be any vector of contiguous indices $(i,i+1,\ldots,i+k-1)$, such that $b_j = x$, for all $i\leq j\leq i+k-1$; here, we call $k$ the \emph{runlength} of the run of the symbol $x$ of interest, starting at position $i$. Next, we let $\rho(\mathbf{b})$ denote the vector of runlengths of runs in $\mathbf{b}$, in the order that the runs appear, and $\iota(\mathbf{b})$ to be the vector of symbols associated with each run, again in the order that the runs appear. For example, if $\mathbf{b} = (1,3,1,1,1,2,2,4)$, we have $\rho(\mathbf{b}) = (1,1,3,2,1)$ and $\iota(\mathbf{b}) = (1,3,1,2,4)$. Further, given the vector of runlengths $\rho(\mathbf{b})$, we define the vector $\rho^{(x)}(\mathbf{b})$ to be the vector of runlengths in $\mathbf{b}$ corresponding to the symbol $x\in \mathcal{X}$, in the order of appearance of the runs. In the previous example, for instance, we will have $\rho^{(1)} = (1,3)$ and $\rho^{(3)} = 1$.

\subsection{Channel Model}
\label{sec:channelmodel}
\begin{figure}[!t]
	\centering
	\resizebox{0.8\linewidth}{!}{
		
		\tikzset{every picture/.style={line width=0.75pt}} 
		
		\begin{tikzpicture}[x=0.75pt,y=0.75pt,yscale=-1,xscale=1]
			
			\draw   (305,96) -- (406,96) -- (406,163) -- (305,163) -- cycle ;
			\draw    (203,131) -- (301,131) ;
			\draw [shift={(304,131)}, rotate = 180] [fill={rgb, 255:red, 0; green, 0; blue, 0 }  ][line width=0.08]  [draw opacity=0] (8.93,-4.29) -- (0,0) -- (8.93,4.29) -- cycle    ;
			\draw    (406,131) -- (574,131) ;
			\draw [shift={(577,131)}, rotate = 180] [fill={rgb, 255:red, 0; green, 0; blue, 0 }  ][line width=0.08]  [draw opacity=0] (8.93,-4.29) -- (0,0) -- (8.93,4.29) -- cycle    ;
			\draw    (354,236) -- (354,166) ;
			\draw [shift={(354,163)}, rotate = 90] [fill={rgb, 255:red, 0; green, 0; blue, 0 }  ][line width=0.08]  [draw opacity=0] (8.93,-4.29) -- (0,0) -- (8.93,4.29) -- cycle    ;
			\draw   (576,96) -- (677,96) -- (677,163) -- (576,163) -- cycle ;
			\draw    (676,131) -- (774,131) ;
			\draw [shift={(777,131)}, rotate = 180] [fill={rgb, 255:red, 0; green, 0; blue, 0 }  ][line width=0.08]  [draw opacity=0] (8.93,-4.29) -- (0,0) -- (8.93,4.29) -- cycle    ;
			
			\draw (318,105) node [anchor=north west][inner sep=0.75pt]   [align=left] {\begin{minipage}[lt]{53.19pt}\setlength\topsep{0pt}
					\begin{center}
						 Duplication\\ channel
					\end{center}
					
			\end{minipage}};
			\draw (215,107.4) node [anchor=north west][inner sep=0.75pt]    {$S_{1} ,\dotsc ,S_{m}$};
			\draw (230,85) node [anchor=north west][inner sep=0.75pt]   [align=left] { $\tau$-mers};
			\draw (233,140.4) node [anchor=north west][inner sep=0.75pt]    {$P_{S |S^-}$};
			\draw (300,241.4) node [anchor=north west][inner sep=0.75pt]    {$K_{1} ,\dotsc ,K_{m} \in \Lambda \setminus \{0\}$};
			\draw (334,265) node [anchor=north west][inner sep=0.75pt]   [align=left] { i.i.d.};
			\draw (448,107.4) node [anchor=north west][inner sep=0.75pt]    {$Z_{1} ,\dotsc ,Z_{T_{m}}$};
			\draw (609,105) node [anchor=north west][inner sep=0.75pt]   [align=left] {\begin{minipage}[lt]{25.96pt}\setlength\topsep{0pt}
					\begin{center}
						DMC
					\end{center}
					
			\end{minipage}};
			\draw (618,135) node [anchor=north west][inner sep=0.75pt]    {$W$};
			\draw (685,107.4) node [anchor=north west][inner sep=0.75pt]    {$Y_{1} ,\dotsc ,Y_{T_{m}}$};
		\end{tikzpicture}
	}
	\caption{The noisy nanopore channel $W_\text{nn}$ }
	\label{fig:nanoporemodel}
\end{figure}

The noisy nanopore channel (NNC), as mentioned earlier, is a noisy duplication channel with an input source that is constrained to have a specific first-order Markov structure. The NNC $W_\text{nn} = W_\text{nn}(\mathcal{X},\mathcal{Y},\tau,P_K,W)$ that we describe here is that introduced in \cite{mcbaininfo1}, with the difference that we assume that the noise arises from a general memoryless channel, and not specifically from an additive white Gaussian noise (AWGN) channel. We shall define each of the parameters of $W_\text{nn}$, below.

Let $\mathcal{X}$ denote the alphabet of possible bases $B$; a natural choice of $\mathcal{X}$ is the set $\{\mathsf{A},\mathsf{T},\mathsf{G},\mathsf{C}\}$ of nucleotides. The input to the channel is a sequence $(S_1,\ldots,S_m)$ of ``states'' or ``$\tau$-mers'', where each $S_i\in \mathcal{X}^\tau$, $i\in [m]$, for some fixed integer $\tau\geq 1$. The integer $\tau$ models the memory (also called ``stationarity") of the nanopore; while small values of $\tau$ lead to a smaller state alphabet and hence more tractable detection algorithms \cite{mcbainsurvey}, we mention that the model in Scrappie assumes that $\tau$ is large.

The $\tau$-mers $S_i$, $i\in [m-1]$, are such that if $S_i = (B_1,\ldots,B_\tau)$, for some (random) bases $B_j\in \mathcal{X}$, $j\in [\tau]$, then its must hold that $S_{i+1} = (B_2,\ldots,B_\tau,B_{\tau+1})$, for some $B_{\tau+1}\in \mathcal{X}$. In other words, the $\tau$-mer $S_{i+1}$ is a left-shifted version of $S_i$, for $i\in [m-1]$. The random process $S^m$ hence is a structured first-order Markov process (or \emph{one-step} Markov process), which we call a \emph{de Bruijn Markov process}. Following \cite{mcbaininfo2}, we assume here that $S^m$ is stationary and ergodic, i.e., irreducible and aperiodic. Let $P_{S|S^-}$ denote its (stationary) transition kernel.

Now, consider an independent and identically distributed (i.i.d.) duplication process $K^m = (K_1,\ldots,K_m)$, with $T_i:= \sum_{j\leq i} K_j$, for $i\in [m]$; here, we let $K_i\stackrel{\text{i.i.d.}}{\sim} P_K$, where the distribution $P_K$ is supported on a set $\Lambda\setminus\{0\}$ of positive integers. We set $T_0:=0$, for convenience. The input sequence $S^m$ is passed through the duplication channel, resulting in the output $Z^{T_m}$, where $(Z_1,\ldots,Z_{T_1}) = (S_1,\ldots,S_1)$, and 
\begin{equation}
(Z_{T_{i-1}+1},\ldots,Z_{T_{i}}) = (S_i,\ldots,S_i),
\end{equation}
for $i\geq 2$. We emphasize that the duplication channel repeats entire $\tau$-mers and not individual bases that constitute a $\tau$-mer. Finally, the sequence $Z^{T_m}$ (of random length) is passed through a memoryless channel $W$, resulting in the final output sequence $Y^{T_m} = (Y_1,\ldots,Y_{T_m})$, where each $Y_j\in \mathcal{Y}$; here, $\mathcal{Y}$ is called the output alphabet. The channel law of the DMC $W$ obeys
\begin{align}
P(Y^{T_m} = \mathbf{y}\mid Z^{T_m} = \mathbf{z}) &= P(Y^{T_m} = \mathbf{y}\mid Z^{T_m} = \mathbf{z}, T_m = \ell(\mathbf{z}))\\ &= \prod_{j=1}^{\ell(\mathbf{z})} W(y_i\mid z_i).
\end{align}
A pictorial depiction of the channel is shown in Fig. \ref{fig:nanoporemodel}. 
\subsection{Channel Capacity}
We define the ergodic-capacity $C(W_\text{nn})$ of $W_\text{nn}$ as the supremum over all rates achievable with vanishing error probability, when the de Bruijn Markov input process $S^m$ is constrained to be stationary and ergodic\footnote{In other words, the rates are obtained via random coding schemes with codewords generated using stationary and ergodic de Bruijn Markov input processes.}. Via techniques from the theory of information stability \cite{han-verdu}, the authors of \cite{mcbaininfo2} establish the equivalence between the above operational definition of capacity and the supremum of a multi-letter mutual information expression. More precisely, the following theorem holds as a corollary of \cite[Thm. 4]{mcbaininfo2}:

\begin{theorem}
	\label{thm:capacity}
	The ergodic-capacity $C(W_\text{\normalfont nn})$ is given by
	\[
	C(W_\text{\normalfont nn}) = \sup_{P_{S|S^-}} \lim_{m\to \infty}\frac{1}{m} I(S^m;Y^{T_m}),
	\]
	where the supremum is over all stationary and ergodic transition kernels $P_{S|S^-}$ of the de Bruijn Markov process $S^m$.
\end{theorem}
\begin{remark}
	Note that for any fixed alphabet $\mathcal{X}$ and $\tau = o(m)$, for any kernel $P_{S|S^-}$, we have that $\lim_{m\to \infty} \frac{1}{m}I(S^m;Y^{T_m}) = \lim_{m\to \infty} \frac{1}{m}I(B^{m+\tau};Y^{T_m})$, where $B^{m+\tau}$ represents the sequence of bases corresponding to $S^m$.
\end{remark}
In what follows, we use the terms ``ergodic-capacity" and ``capacity" interchangeably. Observe that the capacity of a nanopore channel is a multi-letter mutual information expression. Our objective in this work is to derive \emph{explicit, computable} estimates of this expression (via bounds) and novel results pertaining to the ``denoising" of such a channel in parameter regimes that are of immediate practical relevance, with the aid of practical encoding and decoding schemes.
\subsection{Organization of the Paper}
In Section \ref{sec:noiseless}, we first state a (tight) lower bound on the ergodic-capacity of the \emph{noiseless} nanopore channel, which follows from results on rates achieved over duplication channels. Next, in Section \ref{sec:up}, we establish general lower and upper bounds for the setting with noise, via direct manipulations of the mutual information in Theorem \ref{thm:capacity} using information-theoretic inequalities. We then proceed to analyzing the rates achievable over NNCs with erasure noise and large $\tau$-mer lengths in Section \ref{sec:erasure}; specifically, we construct a sequence of $\tau$-mer lengths (which depend on the input lengths) that lead to a ``denoising" of the NNC. In Section \ref{sec:change-point}, we take up the study of general NNCs with high sampling rates (or high numbers of duplications) and discuss a change-point detection-based decoder, which again helps in ``denoising" of the NNC.
\section{Capacity of the Noiseless Nanopore Channel}
\label{sec:noiseless}
In this section, we consider the noiseless nanopore channel $\overline{W}_{\text{nn}}$, which is a special case of $W_{\text{nn}}$ where the DMC $W$ is a ``clean" channel, i.e.,
\begin{equation}
W(y\mid z) = \mathds{1}\{y = z\},
\end{equation}
for all $z\in \mathcal{X}$ and $y\in \mathcal{Y}$. For this special case, we shall derive a single-letter lower bound on the mutual information $I(S^m;Y^{T_m}) = I(S^m;Z^{T_m})$, which will then directly give rise to a lower bound for the ergodic-capacity. 

A proof of the capacity lower bound that we derive can also be obtained after some manipulation from the work in \cite[Thm. 1]{kirsch}, which employs sophisticated tools, but we present a self-contained and simple exposition, which may be of independent interest. Furthermore, while the arguments here only show that the expression we derive is a lower bound on the capacity, the arguments in \cite[Thm. 1]{kirsch} show that this lower bound is in fact the exact expression for $C(\overline{W}_\text{nn})$. We mention however that identifying the exact capacity of $W_\text{nn}$ when $W$ is noisy is a significantly harder problem, primarily because of loss of information about runs of symbols in the input sequence.


To this end, we next introduce some further notation. For any symbol (or base) $b\in \mathcal{X}$, we let $\nu_b(S^m)$ denote the number of runs in $S^m$ corresponding to the $\tau$-mer $^\tau b:= (b,b,\ldots,b)$, and let $\nu(S^m)$ denote the vector $(\nu_b(S^m):\ b\in \mathcal{X})$. For notational ease, we let $G_b\sim P_{G_b}$ denote the common random variable representing the runlengths $\rho^{({^\tau b})}_j:= \left(\rho^{({^\tau b})}(S^m)\right)_j$, $1\leq j\leq \nu_b(S^m)$, corresponding to the $\tau$-mer $^\tau b$, i.e., let $\rho^{({^\tau b})}_j\stackrel{\text{i.i.d.}}{\sim} P_{G_b}$, for all $1\leq j\leq \nu_b(S^m)$.\footnote{The fact that the $\rho^{({^\tau b})}_j$ random variables are i.i.d. follows from the first-order Markov structure of $S^m$.} $G_b$ is a geometric random variable with mean defined to be $1/p_b$, where  $p_b$ is the probability (under $P_{S|S^-}$) of ``leaving" the state (or $\tau$-mer) $^\tau b$ in the corresponding Markov chain. Further, for a fixed (de Bruijn Markov) kernel $P_{S|S^-}$, let $\pi$ denote the stationary distribution of the corresponding Markov chain and let $H(\mathscr{S})$ denote its entropy rate.

Our main result in this section is encapsulated in the following theorem:
\begin{theorem}
	\label{thm:noiseless}
	The ergodic-capacity of the noiseless nanopore channel $\overline{W}_\text{ \normalfont nn}$ is lower bounded as
	\begin{align*}
		C(\overline{W}_\text{\normalfont nn})\geq \max_{P_{S|S^-}} \left[H(\mathscr{S}) - \sum_{b\in \mathcal{X}} \pi(^\tau b) H\bigg(G_b\ \bigg \vert \sum_{j=1}^{G_b} K_j\bigg)\cdot p_b\right].
	\end{align*}
\end{theorem}
\begin{remark}
	Intuitively, for large $\tau$, one expects that optimizing distribution in Theorem \ref{thm:noiseless} should make the stationary probabilities $\pi(^\tau b)$ of $\tau$-mers of the form $^\tau b$ small for all $b\in \mathcal{X}$. This then implies that the capacity $C(\overline{W}_\text{nn})$ should increase to the entropy rate $H(\mathscr{S})$, as $\tau$ increases. This intuition is formalized in Section \ref{sec:no-loop}.
\end{remark}
We reiterate that the lower bound in Theorem \ref{thm:noiseless} is actually tight, following the proof in \cite[Thm. 1]{kirsch}. As a corollary, we obtain the following analytical lower bound on $C(\overline{W}_\text{nn})$:
\begin{corollary}
	\label{cor:noiseless}
	We have that
	\[
	C(\overline{W}_\text{\normalfont nn})\geq 1 - \left(\frac{|\mathcal{X}|-1}{|\mathcal{X}|^{\tau+1}}\right)\cdot \sum_{b\in \mathcal{X}} H\bigg(G_b\ \bigg \vert \sum_{j=1}^{G_b} K_j\bigg).
	\]
\end{corollary}
\begin{proof}
	Consider the case when the probability kernel $P_{S|S^-}$ satisfies ${P}_{S|S^-}(s|s^-) = \frac{1}{|\mathcal{X}|}$, for all admissible ``next" states $s\in \mathcal{X}^\tau$ that are left-shifted versions of the given state $s^-\in \mathcal{X}^\tau$. It can be checked that in this case the stationary distribution $\pi$ is uniform on the state space $\mathcal{X}^\tau$, i.e., $\pi(s) = \frac{1}{|\mathcal{X}|^\tau}$, for all $s\in \mathcal{X}^\tau$, with $p_b = \frac{|\mathcal{X}|-1}{|\mathcal{X}|}$. Plugging in these values into Theorem \ref{thm:noiseless} proves the corollary.
\end{proof}
Before we formally prove Theorem \ref{thm:noiseless}, we discuss some details regarding the computability of the lower bound in the theorem. Clearly, to obtain a computable expression for the capacity lower bound we need to evaluate the conditional entropy term $H(G_b\ \vert \sum_{j=1}^{G_b} K_j)$ for all $b\in \mathcal{X}$, where the random variable $G_b$ is independent from each of the random variables $K_j$ in the conditioning. In what follows, we consider two simple, yet fundamental, duplication channels, and discuss the value of this conditional entropy for those settings. 
\begin{example}[Elementary i.i.d. duplication channel]
	In this setting, each of the (i.i.d.) $K_j$ random variables is of the form $K_j = 1+\text{Ber}(p)$, for some $p\in (0,1)$. Then,
	\begin{equation}
		\sum_{j=1}^{G_b} K_j = G_b+\sum_{j=1}^{G_b} X_j,
	\end{equation}
	where $X_j\sim\text{Ber}(p)$. The summation on the right above corresponds to a ``thinning'' \cite{thinning} of the random variable $G_b$; the thinned random variable is again a geometric random variable $G$ with mean $p/p_b$ (see the discussion after \cite[Example 3]{thinning}), which is independent of $G_b$. The conditional entropy $H(G_b\ \vert \sum_{j=1}^{G_b} K_j)$ can hence be computed as
	\begin{align}
		H\bigg(G_b\ \bigg \vert \sum_{j=1}^{G_b} K_j\bigg) &= H\left(G_b,\sum_{j=1}^{G_b} K_j\right) - H\left(\sum_{j=1}^{G_b} K_j\right) \\
		&= H(G_b)+H\left(\sum_{j=1}^{G_b} K_j\ \bigg \vert\ G_b\right)-H\left(\sum_{j=1}^{G_b} K_j\right) \\
		 &= \frac{h_b(p_b)+h_b(p)}{p_b} - H(G+G_b). \label{eq:entropysum}
	\end{align}
Note that in \eqref{eq:entropysum} above, we have used the fact that $H(\sum_{j=1}^{G_b} K_j \vert G_b) = \E[G_b]\cdot H(K) = \frac{h_b(p)}{p_b}$ and that $H(G_b) = \frac{h_b(p_b)}{p_b}$.


	\label{eg:iid}
\end{example}

\begin{example}[Binomial duplication channel]
	The argument in Example \ref{eg:iid} above can be extended to the case when each $K_j = 1+\text{Bin}(n,p)$, for some $n$. Indeed, one can then write $K_j = 1+\sum_{r=1}^n Y_{j,r}$, where the random variables $Y_{j,r}$ are drawn i.i.d. according to Ber$(p)$. We then obtain, similar to \eqref{eq:entropysum}, that in this case,
	\begin{equation}
	H\bigg(G_b\ \bigg \vert \sum_{j=1}^{G_b} K_j\bigg) = H(G_b)+\frac{H(K)}{p_b}-H(N+G_b), \label{eq:binom}
	\end{equation}
	where $K = \text{Bin}(n,p)$ and $N$ is a negative binomial distribution, independent of $G_b$, with parameters $n$ and $p/p_b$.
\end{example}
The conditional entropies calculated in the above examples can be directly plugged into Corollary \ref{cor:noiseless} to obtain analytical lower bounds on $C(\overline{W}_\text{nn})$. The entropies $H(G+G_b)$ and $H(N+G_b)$ in \eqref{eq:entropysum} and \eqref{eq:binom}, respectively, may be computed numerically, since the PMF of $G+G_b$ (resp. $N+G_b$) is obtained by a simple convolution of the PMFs of $G$ and $G_b$ (resp. $N$ and $G_b$). Nonetheless, simple bounds on such entropies of sums of random variables can be obtained via the inequalities: $\max\{H(X),H(Y)\}\leq H(X+Y)\leq H(X)+H(Y)$, when $X$ and $Y$ are independent random variables.\footnote{Sharper estimates of the entropy of sums above perhaps can be derived using the techniques in \cite{tao-sum,madiman1} and references therein. We do not divert our attention to such directions, for reasons of scope.} Furthermore, we conjecture that it is possible to specialize the result in Theorem \ref{thm:noiseless} for other duplication distributions of interest, using perhaps the results in \cite{entropycheraghchi} for entropy computations or approximations. We now proceed towards a proof of Theorem \ref{thm:noiseless}.

Fix a stationary, ergodic, de Bruijn Markov process $S^m$, with transition kernel $P_{S|S^-}$. We then write
\begin{equation}
	\label{eq:temp1}
I(S^m;Z^{T_m}) = H(S^m)-H(S^m\mid Z^{T_m}).
\end{equation}
The first term $H(S^m)$ on the right side of \eqref{eq:temp1} is easily computable to be $H(S_1)+(m-1)H(S_2\mid S_1)$, with the entropy rate $H(\mathscr{S})$ of the stationary Markov chain with kernel $P_{S|S^-}$ being $H(S_2\mid S_1)$. Hence, our task reduces to explicitly bounding the expression $H(S^m\mid Z^{T_m})$ that is the second term on the right side of \eqref{eq:temp1}. 

The next lemma presents an upper bound on $H(S^m\mid Z^{T_m})$; the essential idea behind its proof is that given the random vector $Z^{T_m}$, the only uncertainty in determining $S^m$ is via the lengths of runs of its symbols. Furthermore, the only ambiguity in the runlengths of symbols in $S^m$, given $Z^{T_m}$, is in those runlengths corresponding to symbols of the form ${}^\tau b$, for some $b\in \mathcal{X}$. This is because \emph{only such symbols} can have runlengths larger than $1$ in $S^m$, owing to the structure of the de Bruijn Markov process. 
\begin{lemma}
	\label{lem:noiseless}
	We have that
	\[
	H(S^m\mid Z^{T_m}) \leq \sum_{b\in \mathcal{X}} \E[\nu_b(S^m)]\cdot H\bigg(G_b\ \bigg \vert \sum_{j=1}^{G_b} K_j\bigg).
	\]
\end{lemma}
\begin{proof}
Observe that
\begin{align}
	H(S^m\mid Z^{T_m})
	&= H(\iota(S^m),\rho(S^m)\mid Z^{T_m}) \\
	&= H(\rho(S^m)\mid Z^{T_m},\iota(S^m)) \label{eq:temp2}\\
	&= H(\rho(S^m)\mid Z^{T_m},\iota(S^m),\nu(S^m)), \label{eq:noiseless2}
\end{align}
where \eqref{eq:temp2} holds since given $Z^{T_m}$, the vector $\iota(S^m)$ is completely determined.

Now, by the structure of the de Bruijn Markov input process $S^m$, the only runs of length larger than $1$ are those that begin with the symbol $^\tau b$, for some $b\in \mathcal{X}$. 

Therefore, given the vector $\iota(S^m)$, the only uncertainty in the vector $\rho(S^m)$ is in the collection \begin{equation}\rho^\text{alike}(S^m) = \left\{\left(\rho^{({^\tau b})}_1,\ldots,\rho^{({^\tau b})}_{\nu_b(S^m)}\right):\ b\in \mathcal{X}\right\}.\end{equation} Hence, continuing from \eqref{eq:noiseless2}, we obtain that
\begin{equation}
	\label{eq:noiseless3}
	H(S^m\mid Z^{T_m}) = H(\rho^\text{alike}(S^m)\mid Z^{T_m},\iota(S^m),\nu(S^m)).
\end{equation}
Now, owing to the Markovity of the process $S^m$, the runlengths in $S^m$ are independent geometric random variables. Hence, from \eqref{eq:noiseless3}, we obtain that
\begin{align}
	H(S^m\mid Z^{T_m})&= H(\rho^\text{alike}(S^m)\mid Z^{T_m},\iota(S^m),\nu(S^m))\\
	&= \sum_{b\in \mathcal{X}} \Pr[\nu_b(S^m) = n_b]\sum_{j=1}^{n_b} H\left(\rho^{({^\tau b})}_j\ \big \vert\ Z^{T_m}, T_m,\iota(S^m)\right)\\
	&\leq \sum_{b\in \mathcal{X}} \Pr[\nu_b(S^m) = n_b]\sum_{j=1}^{n_b} H\left(\rho^{({^\tau b})}_j\ \bigg \vert\ \left(\rho^{(^\tau b)}(Z^{T_m})\right)_j\right)\\
	&= \sum_{b\in \mathcal{X}} \E[\nu_b(S^m)]\cdot H\bigg(G_b\ \bigg \vert \sum_{j=1}^{G_b} K_j\bigg).
\end{align}
Here, the inequality follows since conditioning reduces entropy, and the last equality arises since each of the terms $H(\rho^{({^\tau b})}_j\mid \left(\rho^{(^\tau b)}(Z^{T_m})\right)_j)$ in (c) above equals $H(G_b \vert \sum_{j=1}^{G_b} K_j)$.
\end{proof}
The lemma below presents a computable expression for the quantity $\E[\nu_b(S^m)]$, for any $b\in \mathcal{X}$. Let $(^\tau b)^+$ denote the collection of $\tau$-mers of the form $(b,b,\ldots,b,a_1)$, where $a_1\neq b$, and let $(^\tau b)^-$ denote the collection of $\tau$-mers of the form $(a_2,b,\ldots,b,b)$, where $a_2\neq b$. 
\begin{lemma}
	\label{lem:noiselesshelper}
	For any $b\in \mathcal{X}$, we have
	\begin{align*}
	\E[\nu_b(S^m)]&= (m-1)\pi(^\tau b)\cdot \sum_{s\in (^\tau b)^+} P(s|^\tau b)+ \sum_{s'\in (^\tau b)^-}\pi(s')P(^\tau b|s')\\
	&= (m-1)\pi(^\tau b)\cdot p_b+\sum_{s'\in (^\tau b)^-}\pi(s')P(^\tau b|s').
	\end{align*}
\end{lemma}
\begin{proof}
	The proof follows from the observation that
	\begin{align}
		\nu_b(S^m)&= \sum_{i=1}^{m-1}\mathds{1}\{S_i ={\!} ^\tau b, S_{i+1}\neq{\!} ^\tau b\}+ \mathds{1}\{S_{m-1} \neq{\!}^\tau b, S_m ={\!} ^\tau b\}.
	\end{align}
Employing the linearity of expectation and the structure of the de Bruijn Markov input process $S^m$ gives the statement of the lemma.
\end{proof}
The proof of Theorem \ref{thm:noiseless} is now immediate.
\begin{proof}[Proof of Theorem \ref{thm:noiseless}]
	The proof follows by putting together Lemmas \ref{lem:noiseless} and \ref{lem:noiselesshelper} and then taking a maximum over de Bruijn Markov input processes governed by kernels $P_{S|S^-}$ as in Theorem \ref{thm:capacity}.
\end{proof}

In the next section, we discuss approaches for obtaining bounds on the capacity of the noisy nanopore channel $W_\text{nn}$, i.e., when the DMC $W$ is not clean.
\section{General Bounds on the Capacity of the Noisy Nanopore Channel}
\label{sec:up}
In this section, we present general, computable, lower and upper bounds on $C(W_\text{nn})$, when $W$ is an arbitrary, noisy channel. We first discuss a \emph{lower} bound on the capacity. The intuition behind the bound is that if the duplication process $K^m$ (and hence the ``boundaries" corresponding to each input symbol $S_i$ in the output sequence $Y^{T_m}$) were known, the output symbols corresponding to a fixed input $\tau$-mer can be treated as ``views" through a multi-view channel (see, e.g., \cite{mitzenmacher_datarecovery,landhuber,arnw-t-it}). Before we state our result, we recall some definitions. 


Fix a de Bruijn Markov input process $S^m$, with transition kernel $P_{S|S^-}$, stationary distribution $\pi$, and entropy rate denoted by $H(\mathscr{S})$. It can be checked that the process $Z^{T_m}$ is also stationary, with $\Pr[Z_1 = z] = \pi(z)$. Now, the (scaled) Bhattacharya parameter of the channel $W$ (for the given  kernel $P_{S|S^-}$)  is given by (see, e.g., \cite[Sec. 4.1.2]{sasoglufnt}))
\begin{equation}Z_g(W) := \frac{1}{|\mathcal{X}|^\tau-1}\cdot \sum_{z\neq z'} \sum_{y\in \mathcal{Y}} \sqrt{\pi(z)W(y|z)\pi(z')W(y|z')}.\end{equation}
Now, for any integer $k\geq 1$, let $W^{\otimes k}$ denote the $k$-view DMC $W$, with input alphabet $\mathcal{X}$, output alphabet $\mathcal{Y}^k$, and channel law
\begin{equation}
	W^{\otimes k}(y^k\mid z) = \prod_{i=1}^k W(y_i\mid z).
\end{equation}

\begin{theorem}
	\label{thm:lb}
	We have that
	\[
	C(W_\text{\normalfont nn})\geq \max_{P_{S|S^-}}\left(H(\mathscr{S})- H(K)-\E_K\left[Z_g(W^{\otimes K})\right]\right).
	\]
\end{theorem}
\begin{proof}
Fix the de Bruijn Markov input process $S^m$ with transition kernel ${P}_{S|S^-}$. We first write $I(S^m;Y^{T_m},K^m)$ in two ways:
\begin{equation}
	I(S^m;Y^{T_m},K^m) = I(S^m;Y^{T_m})+I(S^m;K^m\mid Y^{T_m}), \label{eq:help1}
\end{equation}
and
\begin{align}
	I(S^m;Y^{T_m},K^m) &= I(S^m;K^m)+I(S^m;Y^{T_m}\mid K^m)\\
	&=I(S^m;Y^{T_m}\mid K^m), \label{eq:help2}
\end{align}
since $K^m$ is independent of $S^m$. Putting together \eqref{eq:help1} and \eqref{eq:help2}, we obtain that
\begin{align}
	I(S^m;Y^{T_m})&= I(S^m;Y^{T_m}\mid K^m)-I(S^m;K^m\mid Y^{T_m})\\
	&\geq H(S^m)-H(S^m\mid Y^{T_m},K^m)-mH(K), \label{eq:help3}
\end{align}
where the inequality uses the fact that $I(S^m;K^m\mid Y^{T_m})\leq H(K^m) = mH(K)$.  
observe that
\begin{align}
	H(S^m\mid Y^{T_m},K^m)&\stackrel{(a)}{\leq} \sum_{i=1}^m H(S_i\mid Y^{T_m},K^m)\\
	&\stackrel{(b)}{\leq} \sum_{i=1}^m H(S_i\mid Y_{T_{i-1}},\ldots,Y_{T_i},K_i)\\
	&\stackrel{(c)}{\leq} \sum_{i=1}^m \E_K\left[Z_g(W^{\otimes K})\right] = m\cdot\E_K\left[Z_g(W^{\otimes K})\right]. \label{eq:help4}
\end{align}
Here, inequalities (a) and (b) hold since removing the conditioning on some random variables cannot decrease entropy, and (c) holds via \cite[Prop. 4.8]{sasoglufnt}. Hence, we obtain from \eqref{eq:help3} and \eqref{eq:help4} that
\begin{equation}
	\frac{1}{m}I(S^m;Y^{T_m})\geq H(\mathscr{S})- H(K)-\E_K\left[Z_g(W^{\otimes K})\right],
\end{equation}
where $H(\mathscr{S})$ is the entropy rate of the Markov chain with kernel ${P}_{S|S^-}$. The theorem then follows from Theorem \ref{thm:capacity}.
\end{proof}
\begin{remark}
	Theorem \ref{thm:lb} seems to indicate that the uncertainty in the channel $W_\text{nn}$ is primarily due to those in the ``boundaries" of each run of output symbols that arise from the same input symbol (captured by the entropy $H(K)$), and the uncertainty in the estimation of the input symbol from any run of output symbols it gives rise to (captured by the term $\E\left[Z_g(W^{\otimes K})\right]$). In Section \ref{sec:change-point}, we shall discuss a decoding algorithm that makes use of this intuition to ``denoise" the nanopore channel in the regime of large numbers of duplications (or high sampling rates).
\end{remark}
Now, consider the situation when the stationary distribution $\pi$ is uniform on the state space $\mathcal{X}^\tau$ -- this is achieved, for example, by transition probabilities $\overline{P}_{S|S^-}(s|s^-) = \frac{1}{|\mathcal{X}|}$, for all admissible ``next" states $s\in \mathcal{X}^\tau$ for a given state $s^-\in \mathcal{X}^\tau$. For such a setting, the Bhattacharya parameter $Z_g(W)$ evaluates to \begin{equation}\rho(W) := \sum_{z\neq z'}\sum_{y\in \mathcal{Y}} \sqrt{W(y|z)W(y|z')}. 
\end{equation}
This gives rise to the following corollary, obtained by fixing the input process to be $\overline{P}_{S|S^-}$:
\begin{corollary}
	\label{cor:lb}
	We have that
	\[
	C(W_\text{\normalfont nn})\geq \left(1-H(K)-\E\left[\rho(W^{\otimes K})\right]\right).
	\]
\end{corollary}
Evidently, the bound above is non-trivial only when $H(K)< 1-\E\left[\rho(W^{\otimes K})\right].$ Moreover, the following simplification is easy to derive:
\begin{equation}
\label{eq:lb}
\rho(W^{\otimes k}) = \sum_{z\neq z'} \left(\sum_{y\in \mathcal{Y}}\sqrt{W(y|z)W(y|z')}\right)^k.
\end{equation}

As examples, consider the cases when $W$ is a $q$-ary erasure channel or a $q$-ary symmetric channel, where $q = |\mathcal{X}|^\tau$. The $q$-ary erasure channel EC$(\epsilon)$, where $\epsilon\in (0,1)$, has $\mathcal{Y} = \{?\}\cup \mathcal{X}^\tau$, with $W(z|z) = 1-\epsilon$ and $W(?|z) = \epsilon$, for all $z\in \mathcal{X}^\tau$. The $q$-ary symmetric channel, SC$(p)$, where $p\in (0,1)$, has $\mathcal{Y} = \mathcal{X}^\tau$, with $W(z|z) = 1-p$ and $W(y|z) = \frac{p}{|\mathcal{X}|^\tau-1}$, for $y\neq z$, for all $z\in \mathcal{X}^\tau$. Let $W_{\text{nn, EC}}$ and $W_{\text{nn, SC}}$ denote the nanopore channels when $W$ is an erasure channel and a symmetric channel, respectively. Putting together \eqref{eq:lb} and Theorem \ref{thm:lb}, via the symmetry of these channels, it can be derived that
\begin{equation}
	\label{eq:lb-erasure}
	C\left(W_{\text{nn, EC}}\right)\geq 1-H(K)-|\mathcal{X}|^\tau (|\mathcal{X}|^\tau-1)\cdot \E\left[\epsilon^K\right],
\end{equation}
and 
\begin{equation}
	C\left(W_{\text{nn, SC}}\right)\geq 1-H(K)-|\mathcal{X}|^\tau (|\mathcal{X}|^\tau-1)\cdot\E\left[g(p)^K\right],
\end{equation}
where 
\begin{equation}
	g(p):=  2\left(\frac{p(1-p)}{|\mathcal{X}|^\tau-1}\right)^{1/2}+\frac{p(|\mathcal{X}|^\tau-2)}{|\mathcal{X}|^\tau-1}.
\end{equation}
%



As an illustrative example of the performance of our lower bound, let the duplication channel be an elementary i.i.d. duplication channel (see Example \ref{eg:iid}). Figure \ref{fig:erasureplot-lb} shows a plot of the lower bound obtained via \eqref{eq:lb-erasure} for the case when $|\mathcal{X}| = 3$, $\tau = 2$, and the parameter $p = 0.999$ for the i.i.d. duplication channel.\footnote{Note here that, following the intuition in Section \ref{sec:change-point}, we set the duplication parameter $p$ to be high, so as to allow for more ``views" of each input symbol at the decoder, in expectation.}  Evidently, for small memory lengths and base alphabet sizes, and large duplication parameter $p$, the lower bound in Theorem \ref{thm:lb} is quite reasonable. However, these lower bounds are often quite poor for moderate alphabet sizes and memory lengths found in practical nanopore channels (see \cite{mcbaininfo1} for typical values of memory lengths). For example, consider the nanopore channel that is the $W_\text{nn,EC}$, with base alphabet $\mathcal{X} = \{\mathsf{A},\mathsf{T},\mathsf{G},\mathsf{C}\}$ (representing the four DNA bases) and memory length $\tau = 4$.  It can be checked then even for  erasure probabilities as small as $\epsilon \approx 1.3\times 10^{-4}$, the lower bound in \eqref{eq:lb-erasure} turns out to be negative. This provides motivation for the use of alternative methods as in Sections \ref{sec:erasure} and \ref{sec:change-point}.
\begin{figure}[!t]
	\centering
	\subfloat[]{\includegraphics[width = 0.5\linewidth]{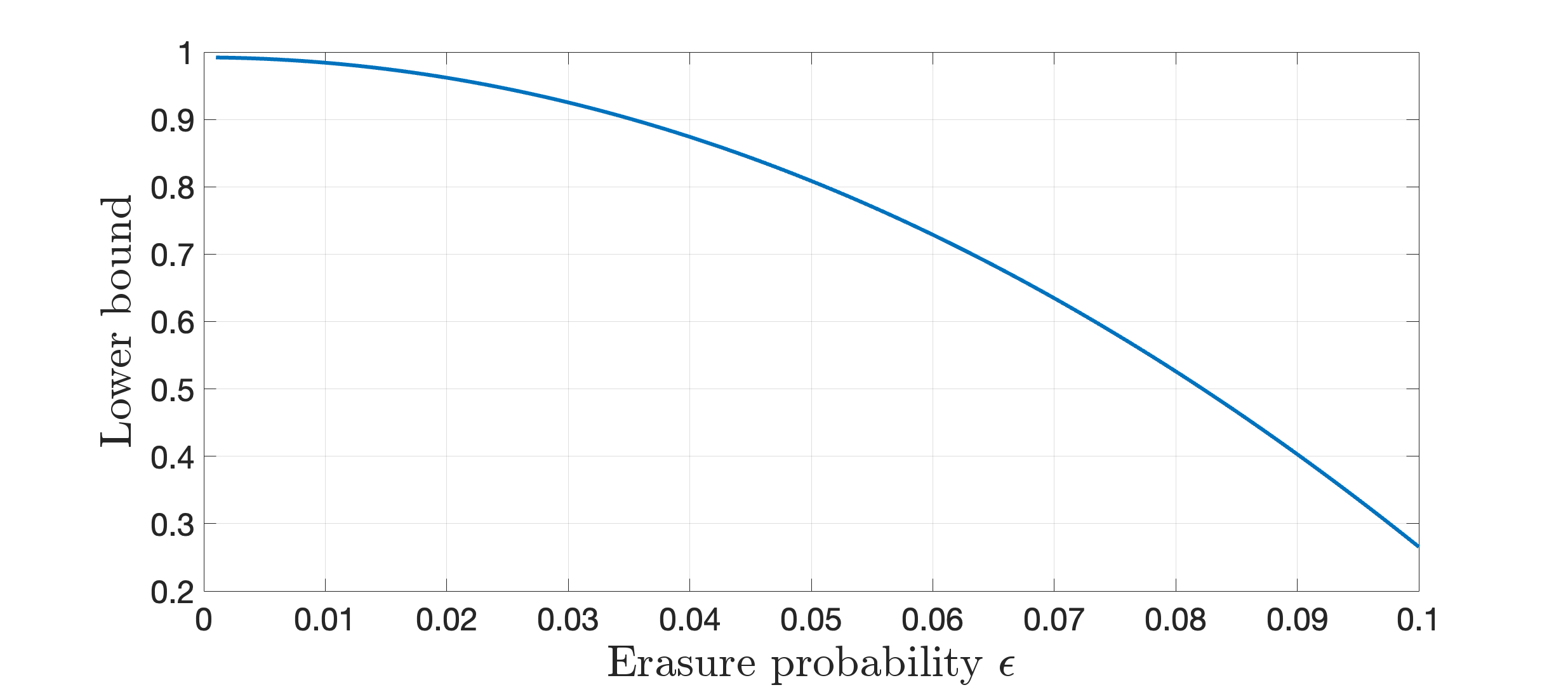} \label{fig:erasureplot-lb}}
	\subfloat[]{\includegraphics[width = 0.5\linewidth]{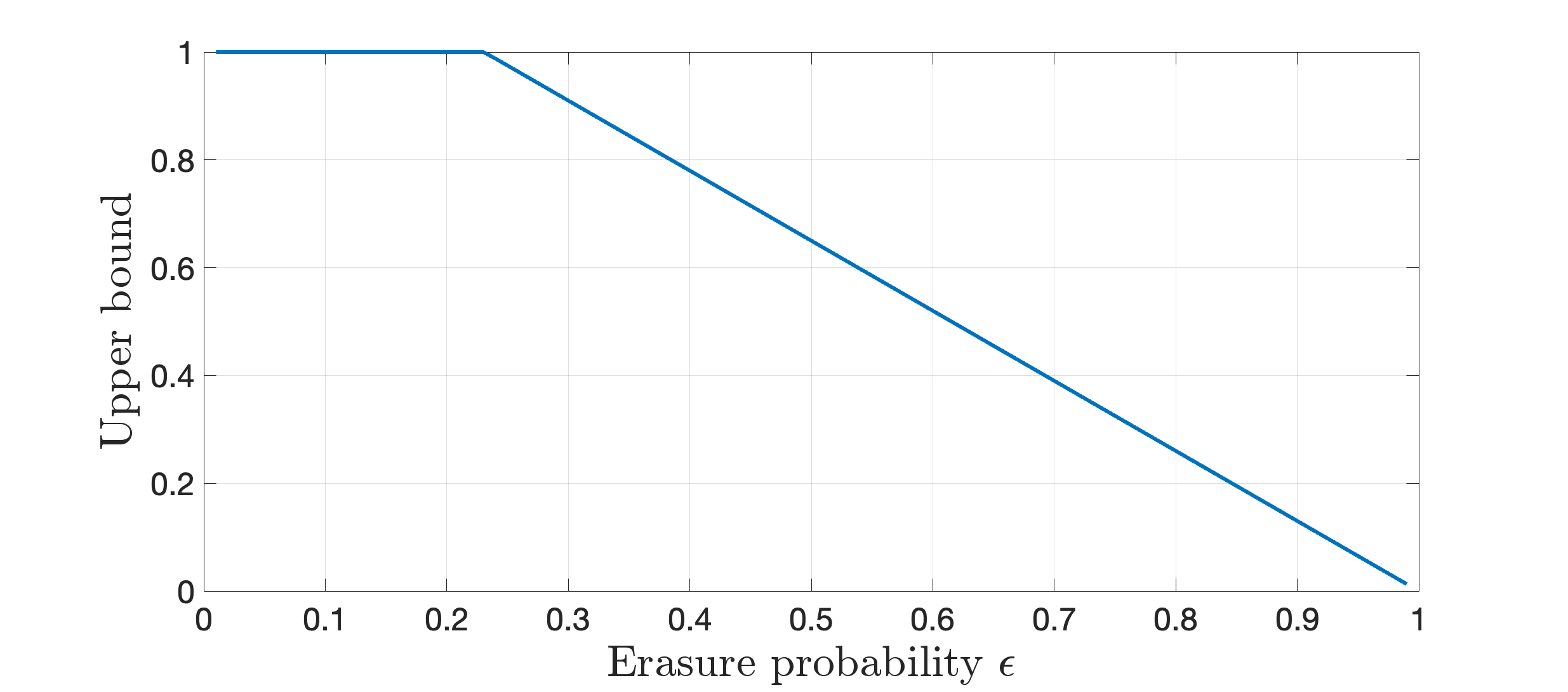} \label{fig:erasureplot-ub}}
	\caption{(a) Our lower bound for $C(W_\text{nn,EC})$, for an i.i.d. duplication channel with parameter $p = 0.999$; (b) Our upper bound for $C(W_\text{nn,EC})$, for an i.i.d. duplication channel with parameter $p = 0.3$. In both cases, we use $|\mathcal{X}| = 3$ and $\tau = 2$.}
	\label{fig:erasureplot}
\end{figure}

We now turn our attention to deriving computable \emph{upper} bounds on $C(W_\text{nn})$. Our first bound is na\"ive, but the second makes use of more structural information about $W_\text{nn}$. We mention that our bounds hold generally for \emph{any} stationary, ergodic Markov input process $P_{S|S^-}$, which is not necessarily a de Bruijn Markov input process.

Let $C(W)$ denote the capacity of the memoryless channel $W$ in the definition of the channel $W_\text{nn}$. Our first bound is as follows. 
\begin{theorem}
	\label{thm:upper1}
	We have that
	\[
	C(W_\text{\normalfont nn})\leq  \E[K]\cdot C(W),
	\]
	where $K\sim P_K$.
\end{theorem}
\begin{proof}
Observe that
\begin{align}
	I(S^m;Y^{T_m}) &\leq I(S^m;Y^{T_m},K^m)\\
	&= I(S^m;K^m)+I(S^m;Y^{T_m}\mid K^m) \\ &= I(S^m;Y^{T_m}\mid K^m), \label{eq:temp3}
\end{align}
where \eqref{eq:temp3} holds due to the independence of the duplication process $K^m$ from the input process $S^m$. Further, we note that $I(S^m;Y^{T_m}\mid K^m) = H(Y^{T_m}\mid K^m) - H(Y^{T_m}\mid K^m,S^m)$. Now, 
\begin{align}
	H(Y^{T_m}\mid K^m,S^m)&= H(Y^{T_m}\mid K^m,S^m, T_m) \\
	 &= \E[T_m]\cdot H(Y\mid Z), \label{eq:up1}
\end{align}
where $Z,Y$ are random variables with the stationary marginal distributions of $Z^{T_m}$ and $Y^{T_m}$, respectively, and where \eqref{eq:up1} follows via Wald's lemma (see, e.g., \cite[Thm. 13.3]{mitzupfal}).

Moreover, similarly,
\begin{equation}
	\label{eq:up2}
	H(Y^{T_m}\mid K^m)\leq H(Y^{T_m}\mid T_m) = \E[T_m]\cdot H(Y). 
\end{equation}
Putting together \eqref{eq:up1} and \eqref{eq:up2} and plugging back into \eqref{eq:temp3} above, we obtain that $I(S^m;Y^{T_m})\leq \E[T_m]\cdot (H(Y)-H(Y\mid Z))$. This then leads to
\begin{align}
	C(W_\text{nn})&\leq \left(\lim_{m\to \infty} \frac{1}{m}\ \E[T_m]\right)\cdot \max\left(H(Y)-H(Y\mid Z)\right) \\
	&= \E[K]\cdot C(W),
\end{align}
since $\E[T_m] = m\E[K]$, where $K\sim P_K$, and using Wald's lemma. 
\end{proof}
For example, for the nanopore channel $W_\text{nn,EC}$ with an i.i.d. duplication channel (see Example \ref{eg:iid}), the upper bound in Theorem \ref{thm:upper1} can be computed as:
\begin{equation}
	\label{eq:erasure-ub}
	C(W_\text{nn,EC})\leq (1+p)\cdot (1-\epsilon),
\end{equation}
which is non-trivial (smaller than $1$) for small $p$ and large $\epsilon$ values, in contrast to the lower bound in \eqref{eq:lb-erasure}, which is non-trivial (larger than $0$) for large $p$ and small $\epsilon$ values. Figure \ref{fig:erasureplot-ub} shows a plot of the upper bound obtained via \eqref{eq:erasure-ub} for the case when $|\mathcal{X}| = 3$, $\tau = 2$, and the parameter $p$ when the duplication channel is the i.i.d. duplication channel (see Example \ref{eg:iid}) equals $0.3$. 

While the upper bound above relates the capacity of $W_\text{nn}$ with the capacity of the DMC $W$, a simple upper bound (via the data processing inequality \cite[Thm. 2.8.1]{cover_thomas}) on $C(W_\text{nn})$ is:
\begin{equation}
	\label{eq:upper2}
C(W_\text{nn})\leq C(\overline{W}_\text{nn}),
\end{equation}
i.e., $C(W_\text{nn})$ is upper bounded by the capacity of the noiseless nanopore channel. 

In the following theorem, we present an upper bound on $C(W_\text{nn})$ that, for selected duplication distributions $P_K$ and DMCs $W$, improves on the bound in \eqref{eq:upper2}. For any fixed $P_{S|S^-}$, let $Z,Y$ be random variables with the stationary marginal distributions of $Z^{T_m}$ and $Y^{T_m}$, respectively.

\begin{theorem}
	\label{thm:upper2}
	We have that
	\begin{align*}
	C(W_\text{\normalfont nn})\leq \max_{P_{S|S^-}}\left[\lim_{m\to \infty} \frac{1}{m} I(S^m;Z^{T_m}) -(\E[K]\cdot H(Z\mid Y)-H(K))\right].
	\end{align*}
\end{theorem}
For distributions $P_K$ with low entropy, but relatively high expected value in comparison, the upper bound in Theorem \ref{thm:upper2} is likely to be tighter than that in \eqref{eq:upper2}.
\begin{proof}
	For any input transition kernel $P_{S|S^-}$, we write
	\begin{equation}
		\label{eq:up3}
		I(S^m;Y^{T_m}) = H(S^m)-H(S^m\mid Y^{T_m}).
	\end{equation}
	Now, note that 
	\begin{align}
		H(S^m\mid Y^{T_m}) &= H(S^m,Z^{T_m}\mid Y^{T_m}) - H(Z^{T_m}\mid S^m,Y^{T_m})\\ &\stackrel{(a)}{\geq} H(S^m,Z^{T_m}\mid Y^{T_m})  - H(Z^{T_m}\mid S^m) \\&= H(Z^{T_m}\mid Y^{T_m})+H(S^m\mid Z^{T_m},Y^{T_m}) - H(Z^{T_m}\mid S^m)\\
		&\stackrel{(b)}{=} H(Z^{T_m}\mid Y^{T_m})+H(S^m\mid Z^{T_m}) - mH(K)\\
		&\stackrel{(c)}{=} m(\E[K]\cdot H(Z\mid Y)-H(K)) + H(S^m\mid Z^{T_m}),
	\end{align}
where $(a)$ holds since conditioning reduces the entropy and $(b)$ holds since $S^m$ -- $Z^{T_m}$ -- $Y^{T_m}$ forms a Markov chain. Finally, $(c)$ follows from the fact that $H(Z^{T_m}\mid Y^{T_m}) = \E[T_m]\cdot H(Z\mid Y)$, by the memorylessness of the channel $W$; further, we have $\E[T_m] = m\E[K]$. Plugging into \eqref{eq:up3} gives us the theorem.
\end{proof}
\begin{remark}
The limiting mutual information rate $\lim_{m\to \infty} \frac{1}{m} I(S^m;Z^{T_m})$, for any fixed $P_{S|S^-}$ can be computed from Lemmas \ref{lem:noiseless}  and \ref{lem:noiselesshelper}, given the tightness of the bound in Lemma \ref{lem:noiseless}, as proved in \cite{kirsch}. Further, the joint distribution $P_{Z,Y}$ obeys	$P_{Z,Y}(z,y) = \pi(z)W(y|z)$, 
which allows for a computation of $H(Z\mid Y)$ in Theorem \ref{thm:upper2}. 
\end{remark}
Until now, we have presented general bounds on the capacity of the nanopore channel, with arbitrary duplication and noise processes, and for any value of memory length $\tau$. However, these bounds tend to be somewhat poor for regimes of practical interest, which require large $\tau$-mer lengths. To address this issue, in the next section, we shall focus on the specific class of NNCs with erasure noise, with long $\tau$-mer lengths. Our results show, somewhat surprisingly, that in the limit of large $\tau$, rates arbitrarily close to $1$ can be achieved over this channel, using low-complexity encoding and decoding schemes.

\section{Achievable Rates Over NNCs with Erasure Noise For Long $\tau$-mer Lengths}
\label{sec:erasure}
In this section, we focus on the special case when $W$ is an erasure channel EC$(\epsilon)$, where $\epsilon\in (0,1)$ (see the discussion following Theorem \ref{thm:lb} for a definition). Thus, our nanopore channel is $W_\text{nn, EC} = W_\text{nn}(\mathcal{X},\mathcal{X}^\tau\cup \{?\},\tau,P_K,W)$. To make the dependence on the ``memory" of the nanopore explicit, we write $C^{(\tau)}(W_\text{nn, EC})$ as the capacity of the NNC of interest. 

Our main objective in this section is to show that for an NNC with erasure noise, one can achieve rates arbitrarily close to $1$, so long as $\tau$ is large enough. Recall that all through, we assume that the quantities $|\mathcal{X}|$ and $|\Lambda|$ are fixed (recall that $\Lambda\setminus \{0\}$ is the support set of the duplication distribution $P_K$). Before we present the theorem statement, we need some more definitions. 

For every $\tau\geq 1$, we define a specific input process $P^{\star}_{S|S^-} = P_{S|S^-}^{\star, (\tau)}$, which is a maximal entropy de Bruijn Markov input process, under the constraint that it \emph{does not have self-loops} on any of its states $s\in \mathcal{X}^\tau$. Note that the only possible self-loops in a general de Bruijn Markov input process are on states (or symbols) of the form ${}^\tau b$, for some $b\in \mathcal{X}$. The class of ``no-self-loop" de Bruijn Markov processes that we consider eliminates such self-loops.

We remark that if $S^\text{no-loop}_\tau$ denotes the constrained system that consists of sequences generated by de Bruijn Markov input processes on $\mathcal{X}^\tau$ with no self-loops, then the code generated by $P^{\star}_{S|S^-}$ has the largest rate $C^\text{no-noise, no-loop}_\tau$, which is also called the (noiseless) capacity of $S^\text{no-loop}_\tau$ (see \cite[Chap. 3]{Roth} for more details).\footnote{Furthermore, for a Markov chain generated using $P^{\star}_{S|S^-}$, its entropy rate equals $C^\text{no-noise, no-loop}_\tau$ (see \cite[Thm. 3.23]{Roth}).} Likewise, we let $C^\text{no-noise}_\tau$ denote the noiseless capacity of the constrained system $S_\tau$ consisting of sequences generated by any de Bruijn Markov input process.


Our main result in this section is summarized in the theorem below:

\begin{theorem}
	\label{thm:erasuremain}
	We have that 
	$$
	{C}^{(\tau)}(W_\text{\normalfont nn, EC}) \geq C_\tau^\text{\normalfont no-noise, no-loop} -O(\tau\epsilon^{\tau}).
	$$
	Hence, $\lim_{\tau\to \infty} {C}^{(\tau)}(W_\text{\normalfont nn, EC}) = 1$.
\end{theorem}
Interestingly, the theorem above suggests that longer memory lengths $\tau$ give rise to higher capacities of the associated NNC with erasure noise. Intuitively, such a result arises because larger $\tau$ values imply that in the de Bruijn Markov input process, longer lengths of paths are with high probability uniquely determined by their endpoints, thereby allowing for longer bursts of erasures.

The proof of Theorem \ref{thm:erasuremain} relies on obtaining a lower bound on $ \overline{C}(W_\text{\normalfont nn, EC})$, via the specific class of no-self-loop de Bruijn Markov input processes above. In what follows, we collect some useful facts about this class of input processes.

\subsection{Properties of de Bruijn Markov Processes With No Self-Loops}
\label{sec:no-loop}
The main attribute of no-self-loop de Bruijn Markov input processes that is useful for our analysis is that all sequences (or codewords) generated by such processes are such that any $\tau$-mer in the codeword has runs of length only $1$, if it occurs. A first observation about codewords generated by such a process is presented as a lemma below (the straightforward proof is omitted). 

\begin{lemma}
	\label{lem:noloop}
	For any codeword $\mathbf{c} = (c_1,\ldots,c_n)$ generated by a no-self-loop de Bruijn Markov input process, we have that for all $i, j\in [n]$ such that $i\leq j\leq i+\tau$, the symbols $c_i$ and $c_{j}$ completely determine $c_{i+1},\ldots,c_{j-1}$.
\end{lemma}

We next state a useful fact about the noiseless capacities $C_\tau^\text{\normalfont no-noise, no-loop}$ and $C_\tau^\text{\normalfont no-noise}$.

\begin{theorem}
	\label{thm:nonoisecap}
	We have that 
	\[
	\lim_{\tau \to \infty} C_\tau^\text{\normalfont no-noise, no-loop} = \lim_{\tau \to \infty} C_\tau^\text{\normalfont no-noise} = 1.
	\]
\end{theorem}

Before we prove Theorem \ref{thm:nonoisecap}, we need additional notation and a helpful lemma. For any fixed $\tau\geq 1$, let $G_\tau$ denote the (irreducible, lossless) graph that presents $S$ (see \cite[Ch. 2]{Roth} for definitions). Further, let $A_{G_\tau}$ denote the $|\mathcal{X}|^\tau\times |\mathcal{X}|^\tau$ adjacency matrix of $G_\tau$. Likewise, let $G^\text{no-loop}_\tau$ denote the graph presenting the constrained system $S^\text{no-loop}_\tau$, and let $A_{G_\tau^\text{no-loop}}$ denote its adjacency matrix. For any square matrix $B$, let $\lambda(B)$ denote its largest eigenvalue.

The following lemma, proved in Appendix \ref{app:eig}, then holds.
\begin{lemma}
	\label{lem:eig}
	For any $\tau\geq 2$, we have that
	\[
	\lambda(A_{G_{\tau-1}^\text{\normalfont no-loop}})< \lambda(A_{G_{\tau}^\text{\normalfont no-loop}}).
	\]
\end{lemma}
With this lemma in place, we are in a position to prove Theorem \ref{thm:nonoisecap}. The proof relies on key ideas in the theory of constrained systems (see \cite{Roth} for more details). 

\begin{proof}[Proof of Thm. \ref{thm:nonoisecap}]
	By the definition of a de Bruijn Markov process, we see that each row of $A_{G_\tau}$ consists of $|\mathcal{X}|$ $1$s and $|\mathcal{X}|^\tau - |\mathcal{X}|$ $0$s. From \cite[Thm.  3.23]{Roth}, we have that $C_\tau^\text{no-noise} = \log_{|\mathcal{X}|}(\lambda(A_{G_\tau}))$, where $\lambda(A_{G_\tau})$ is the largest eigenvalue of $A_{G_\tau}$. By standard arguments (see, e.g., \cite[Prop. 3.14]{Roth}), we have that $\lambda(A_{G_\tau}) = |\mathcal{X}|$, implying that for any $\tau\geq 1$, we have $C_\tau^\text{no-noise} = 1$.
	
	Hence, it remains to be shown that $\lim_{\tau \to \infty} C_\tau^\text{\normalfont no-noise, no-loop} =1$. Now, observe that those rows of $A_{G_\tau^\text{no-loop}}$ corresponding to a state of the form $^\tau b$, for some $b\in \mathcal{X}$, have exactly $|\mathcal{X}|-1$ $1$s, and all other rows have exactly $|\mathcal{X}|$ $1$s. Again, from \cite[Prop. 3.14]{Roth} and \cite[Thm. 3.23]{Roth}, we see that \begin{equation}\log_{|\mathcal{X}|}(|\mathcal{X}|-1)\leq C_\tau^\text{\normalfont no-noise, no-loop} = \log_{|\mathcal{X}|}(\lambda(A_{G_\tau^\text{no-loop}}))\leq 1.\end{equation}
	From Lemma \ref{lem:eig}, we see that $\lambda(A_{G_\tau^\text{no-loop}})$ strictly increases with $\tau$, thereby proving the theorem.
\end{proof}

In the next section, we prove Theorem \ref{thm:erasuremain}.
\subsection{Proof of Theorem \ref{thm:erasuremain}}
\label{sec:erasure-b}


Recall that we employ the input process $P^\star_{S|S^-}$, which is a no-self-loop de Bruijn Markov process. The proof of Theorem \ref{thm:erasuremain} relies on the use of a low-complexity decoder, $\mathscr{D}_\text{clean}$, which proceeds as follows. The decoder $\mathscr{D}_\text{clean}$ replaces all bursts of erasures of length $\tau-1$ or less with the collection of true symbols that were erased, in the same order, but repeated arbitrarily so that the length of the decoded burst of erasures equals the length of the burst itself. Indeed, observe from Lemma \ref{lem:noloop} that given the symbols immediately before the start and after the end of such a burst of erasures, the actual sequence of $\tau$-mers in the transmitted codeword corresponding to the burst can be decoded. The decoder $\mathscr{D}_\text{clean}$ then repeats each symbol in this actual sequence of $\tau$-mers arbitrarily, so that the length of the decoded output equals the length of the burst. 

The next proposition establishes a helpful identity for the case when the input process of the nanopore channel is $P_{S|S^-}^{\star,(\tau)}$.
Let $I_{P^{\star, (\tilde{\tau})}}(S^{m};Y^{T_m})$ denote the mutual information between $S^m$ and $Y^{T_m}$ when the input process is $P^{\star, ({\tau})}_{S|S^-}$. 
\begin{proposition}
	\label{prop:erasure}
	We have that
	$$\lim_{m\to \infty}\frac{1}{m}I_{P^{\star, (\tau)}}(S^m;Y^{T_m})\geq C_\tau^\text{\normalfont no-noise, no-loop} - \E[K]\epsilon^\tau\cdot \tau\log|\mathcal{X}|.$$
\end{proposition}
The proof of Proposition \ref{prop:erasure} requires a helpful lemma. 
Let
\begin{align}
	\mathcal{E}:= \left\{S^m\neq f(Y^{T_m})\right\},
\end{align}
where $f: \mathcal{Y}^*\to \left(\mathcal{X}^\tau\right)^m$ is the MAP estimator of $S^m$ given $Y^{T_m}$. It is well known that $f_\ell$ has the lowest error probability among all possible estimators of $S^m$ given $Y^{T_m}$.
The following lemma then holds.

\begin{lemma}
	\label{lem:erasure-a}
	We have that 
	$
	\Pr[\mathcal{E}] \leq m\E[K]\cdot \epsilon^\tau.
	$
\end{lemma}
\begin{proof}
	Let $\overline{\mathcal{E}}$ denote the following event: \begin{equation}\overline{\mathcal{E}}:= \{\text{Some burst of erasures in $Y^{T_m}$ has length at least $\tau$}\}.\end{equation}
	Note that the probability that a given burst of erasures has length at least $\tau$, equals $\epsilon^\tau$. From the structure of the no-self-loop Markov input process, we see from Lemma \ref{lem:noloop} that if the event $\overline{\mathcal{E}}$ \emph{does not hold} for any $\ell\in [m]$, then the sequence $S^m$ is exactly reconstructible from $Y^{T_m}$ via $\mathscr{D}_\text{clean}$, and hence by the MAP decoder $f$. Thus, we have that
	\begin{align}
		\Pr[\mathcal{E}]&\leq \Pr[\overline{\mathcal{E}}]\\
		&= \E[\Pr[\mathcal{E}\mid T_m]]\\
		&\leq \E[T_m\cdot \epsilon^\tau]= m\E[K]\cdot \epsilon^\tau.
	\end{align}
	Here, the second inequality is via a union bound on the probability of a burst of erasures of length at least $\tau$, starting at some index $i\in [T_m]$, for fixed $T_m$. The statement of the proposition then follows readily.
\end{proof}
Lemma \ref{lem:erasure-a} then affords a proof of Proposition  \ref{prop:erasure}.
\begin{proof}
	Fix the input distribution $P^{\star, (\tau)}_{S|S^-}$. We then have that
	\begin{align}
		\frac{1}{m}I_{P^{\star, (\tau)}}(S^m;Y^{T_m})
		&=\frac{1}{m}\left[H(S_1)+(m-1)H(S_2\mid S_1)-H(S^m\mid Y^{T_m})\right]\\
		&\geq \frac{1}{m}\bigg[(m-1)H(S_2\mid S_1)\ -  \bigg( 1+\Pr[\mathcal{E}]\cdot \tau \log |\mathcal{X}|\bigg)\bigg]\\
		&\geq \frac{1}{m}\bigg[(m-1)H(S_2\mid S_1)\ -  \bigg( 1+m\E[K]\epsilon^\tau\cdot \tau \log |\mathcal{X}|\bigg)\bigg].
	\end{align}
	Taking the limit as $m\to \infty$, we get that
	\begin{align}
		\lim_{m\to \infty}\frac{1}{m}I_{P^{\star, (\tau)}}(S^m;Y^{T_m})
		&\geq H(S_2|S_1) - \E[K]\epsilon^\tau\cdot \tau\log|\mathcal{X}|\\
		&= C_\tau^\text{no-noise, no-loop} - \E[K]\epsilon^\tau\cdot \tau\log|\mathcal{X}|,
	\end{align}
	thereby proving the proposition. Here, the last equality follows from the fact that $P^{\star, (\tau)}_{S|S^-}$ has the maximal entropy, i.e., achieves the noiseless capacity of the constraint $S_\tau^{\text{no-loop}}$ (see \cite[Thm. 3.23]{Roth}).
	%
\end{proof}
The proof of Theorem \ref{thm:erasuremain} then follows immediately.
\begin{proof}[Proof of Thm. \ref{thm:erasuremain}]
		Via Proposition \ref{prop:erasure}, we see that
	\begin{align}
		{C}^{(\tau)}(W_\text{\normalfont nn, EC})
		&\geq  \lim_{m\to \infty} \frac{1}{m} I_{P^{\star, (\tau)}}(S^m;Y^{T_m})\\
		&\geq C_\tau^\text{\normalfont no-noise, no-loop} - O(\tau\epsilon^\tau).
	\end{align}
	The proof that $\lim_{\tau\to \infty} {C}^{(\tau)}(W_\text{\normalfont nn, EC}) = 1$ then follows via Theorem \ref{thm:nonoisecap}, using the trivial observation that ${C}^{(\tau)}(W_\text{\normalfont nn, EC})\leq 1$, for all $\tau\geq 1$.
\end{proof}
In the next section, we shall focus on an interesting regime of operation of NNCs with \emph{arbitrary} (but regular) noise distributions, which is that when the sampling rates are chosen to be high, so as to give rise to large numbers of $\tau$-mer duplications. Once again, we shall see that interestingly, rates arbitrarily close to $1$ can be achieved in this setting, using practical encoding and decoding schemes.
\section{A Change-Point Detection-Based Decoder for High Sampling Rates}
\label{sec:change-point}
In this section, we propose a decoding algorithm for general nanopore channels $W_\text{nn}$, for the case when the rates of measurements of the electric currents at the end of the nanopore channel (also called ``sampling rates") are high. High sampling rates give rise to duplication random variables $K_i$, $i\in [m]$, which are typically high. In addition, the work \cite[Sec. II-B]{mao} also mentions the possibility of using change-point detection algorithms such as those employed in practice \cite{laszlo} for ``finding the transitions of the dwelling" $\tau$-mers. 

Fix an arbitrary no-self-loop de Bruijn Markov input process $P_{S|S^-}$ and a general nanopore channel $W_\text{nn}(\mathcal{X},\mathcal{Y},\tau,P_K,W)$. Assume, in addition, the natural regularity condition that the distributions $W_{Y|z}$ and $W_{Y|z'}$ are not identical, for any pair $z\neq z'$. 

Our decoding algorithm consists of two stages. In the first stage, an optimal change-point detection algorithm (see, e.g., \cite{veeravalli} for details on quickest change detection) is employed for estimating the time intervals $T_i$, $i\in [m]$, which form the ``boundaries" of the run of output symbols that arise from a single input symbol. Note that the time intervals $T_i$, $i\in [m]$, are indeed change-points, since the distribution of output symbols changes from $W_{Y|z}$ to $W_{Y|z'}$, for some  $z, z'\in \mathcal{X}^\tau$ with $z'\neq z$. After suitable processing of the estimates from the first stage, the second stage of our algorithm performs optimal (maximum aposteriori, or MAP) decoding on the output symbols within each estimated boundary, to decode the corresponding input symbol. The intuition is that if the estimates produced in the first stage are fairly accurate, then in the setting of high sampling rates, there are sufficiently many samples within each boundary to decode each input symbol correctly with high probability.

Fix a length $m\geq 1$ of the input sequence $S_1,\ldots,S_m$. Let $\ell_m:= m^2(\ln m)^3$ and $h_m:=\gamma m^2(\ln m)^3$, for some $\gamma >1$. We set  the sampling rates high enough so that 
\begin{equation}
	\label{eq:sampling}
	P_K(\ell_m \leq K\leq h_m)\geq 1-\frac{1}{m^{1+\eta}},
\end{equation}
for some $\eta>0$. Clearly, this implies via a union bound that
\begin{equation}
	\label{eq:e0}
	\lim_{m\to \infty} \Pr\left[\ell_m\leq K_i\leq h_m,\ \text{for all}\ i\in [m]\right] = 1.
\end{equation}
Further, set a false alarm probability $\alpha_m:= \frac{1}{m^3(\ln m)^4}$ and a ``trimming length" $c_m:=m(\ln m)^2$. 
\begin{algorithm}[t]
	\caption{A decoder for high sampling rates}
	\label{alg:decode}
	\begin{algorithmic}[1]	
		\Procedure{\textsc{Decode}}{$y^{t_m}$}
		\State Set \texttt{start} $\gets 1$, \texttt{end} $\gets 1$, and $j\gets 1$.
		\While{\texttt{end}$< t_m$}
			\State Compute a change-point estimate $\widehat{t}_j$ using the Shiryaev algorithm \cite{shiryaev1}, \cite[Alg. 1]{veeravalli} on samples $y_i$, $i\geq \texttt{start}$, with $\alpha_m$ as input. \label{step:boundary}
			\State Decode $\widehat{s}_{j}\gets \textsc{MAP}(y_\texttt{start},\ldots,y_{\widehat{t}_j-c_m})$. \label{step:decode}
			\State Update $j\gets j+1$ and \texttt{start} $\gets \widehat{t}_j+1$.
		\EndWhile
		\State Return $(\widehat{s}_1,\ldots,\widehat{s}_j)$.
		\EndProcedure	
	\end{algorithmic}
\end{algorithm} 

Our decoding algorithm, shown as Algorithm \ref{alg:decode}, acts on any given instance $y^{t_m}$ of the output sequence $Y^{T_m}$. It consists of two stages:
\begin{enumerate}
	\item The first stage, shown as Step \ref{step:boundary}, uses the well-known, optimal change-point detection algorithm (for Bayesian quickest change detection) that is the Shiryaev algorithm \cite{shiryaev1}, \cite[Alg. 1]{veeravalli}, which on input of the false alarm probability, computes estimates of the intervals $T_i$, $i\in [m]$, sequentially.
	\item The second stage, shown as Step \ref{step:decode}, first uses the trimming interval to throw away the last $c_m$ of the samples in $y_\texttt{start},\ldots,y_{\widehat{t}_j}$, in each iteration. The remaining samples are treated as noisy views \cite{arnw-t-it} of a single input symbol via the DMC $W$, which are then decoded to a single symbol $\widehat{s}_j\in \mathcal{X}^\tau$, using the optimal MAP decoder.
\end{enumerate}
We now proceed to analyze the performance of our decoding algorithm. The following well-known result \cite[Thm. 3.2]{veeravalli}, \cite{tartakovsky-veer} will be useful to us. 
\begin{theorem}
	\label{thm:change}
	Let $\{X_i\}_{i\geq 1}$ be an i.i.d. sequence of random variables such that $X_1,\ldots,X_K\sim P_0$ and $X_{K+1},\ldots \sim P_1$, for some unknown, random $K\sim P_K$. Then, for any $a_m\xrightarrow{m\to \infty}0$, the change-point estimate $K_s$ returned by the Shiryaev algorithm has false alarm probability $\Pr[K_s<K]\leq a_m$, with
	\[
	\E[\max\{K_s-K,0\}]\leq \frac{-\ln \alpha}{D(P_1||P_0)}\cdot (1+\delta),
	\]
	for any $\delta>0$.
\end{theorem}

Our claim is captured in the following theorem.
\begin{theorem}
	\label{thm:decodemain}
	For any no-self-loop de Bruijn Markov input process $P_{S|S^-}$, in the high sampling-rate regime \eqref{eq:sampling}, we have
	\[
	\lim_{m\to \infty}\Pr\left[\textsc{Decode}(Y^{T_m})\neq S^m\right] = 0.
	\]
\end{theorem}
Towards proving Theorem \ref{thm:decodemain}, we define the following error events. Let
\begin{align}
	\mathcal{E}_1&:= \{K_i>h_m\ \text{or}\ K_i<\ell_m,\ \text{for some } i\in [m]\},\\
	\mathcal{E}_2&:= \{\text{A false alarm occurs for some $T_i$, $i\in [m]$}\},\\
	\mathcal{E}_3&:=\{\text{Detection delay for $T_i$ is larger than $c_m$,}\   \text{for some $i\in [m]$}\},\\
	\mathcal{E}_4&:=\{\text{MAP decoder decodes some $S_i$, $i\in [m]$, incorrectly}\}.
\end{align}
We now prove Theorem \ref{thm:decodemain}.
\begin{proof}
	The events $\mathcal{E}_i$, $i\in [4]$ constitute the error events for the decoder in Algorithm \ref{alg:decode}, in that $\{\textsc{Decode}(Y^{T_m})\neq S^m\} \subseteq \cup_{i=1}^4 \mathcal{E}_i$, and our proof shows that $\lim_{m\to \infty} \Pr\left[\bigcap_{i=1}^4 \mathcal{E}_i^c\right] = 1$. Fix a sufficiently large length $m$ of the input sequence. First, from \eqref{eq:sampling}, we see that $\Pr[\mathcal{E}_1^c] \geq 1-\frac{1}{m^\eta} =: 1-\zeta_{1,m}$.
	
	Next, consider $\Pr\left[\mathcal{E}_2^c\mid \mathcal{E}_1^c\right]$. Via a union bound, conditioned on the event $\{K_i\leq h_m,\ \text{for all }i\in [m]\}$, we have that 
	\begin{align}
		\Pr\left[\mathcal{E}_1^c\mid \mathcal{E}_0^c\right]&\geq 1-mh_m\cdot \alpha_m\\
		&=1-\frac{\gamma}{\ln m}:=1-\zeta_{2,m}. \label{eq:dec1}
	\end{align}
	Now, consider $\Pr\left[\mathcal{E}_3^c\mid \mathcal{E}_1^c, \mathcal{E}_2^c\right]$. By conditioning on $\mathcal{E}_1^c$, we cannot have false alarms in the detection of \emph{any} of the intervals $T_i$, $i\in [m]$. This implies that the number of iterations of the loop in Algorithm \ref{alg:decode} is at most $m$. Now, via Theorem \ref{thm:change} and an application of the Markov inequality, we see that for any $i\in [m]$, if $\widehat{T}_i$ is the estimate of $T_i$ returned by Algorithm \ref{alg:decode}, 
	\begin{equation}
		\Pr\left[\max\{\widehat{T}_i-T_i,0\}\geq c_m\right]\leq \frac{-(1+\delta)\cdot \ln \alpha}{c_m\cdot \min_{z\neq z'} D(W_{Y|z}||W_{Y|z'})},
	\end{equation}
where $\delta>0$ is some fixed constant. Hence, via a union bound, we have
\begin{align}
		\Pr\left[\max\{\widehat{T}_i-T_i,0\}\geq c_m,\ \text{for some }i\in [m]\right]&\leq \frac{- (1+\delta)m\cdot \ln \alpha_m}{c_m\cdot \min_{z\neq z'} D(W_{Y|z}||W_{Y|z'})}\\
		&\leq  \frac{rm\ln m}{m(\ln m)^2} = \frac{r}{\ln m},
\end{align}
for some absolute constant $r>0$. Hence, we have that \begin{equation}\Pr\left[\mathcal{E}_3^c\mid \mathcal{E}_1^c, \mathcal{E}_2^c\right]\geq 1-\frac{r}{\ln m}:=1-\zeta_{3,m}.\end{equation}
Finally, consider the probability $\Pr\left[\mathcal{E}_4^c\mid \mathcal{E}_1^c, \mathcal{E}_2^c, \mathcal{E}_3^c\right]$. Note now that conditioned on $\mathcal{E}_1^c$ and $\mathcal{E}_3^c$, since $\ell_m-mc_m>0$, there are exactly $m$ ``boundaries" (including the boundary at $T_m$) estimated by the change-point detection procedure. Further, the length of each of these boundaries is at least $\ell_m-mc_m =  m^2(\ln m)^3-m^2(\ln m)^2\geq m^2 (\ln m)^2:=\iota_m$, for sufficiently large $m$. Hence, by a union bound, using \cite[Prop. 4.7]{sasoglufnt}, we have that
\begin{align}
	\Pr\left[\mathcal{E}_4^c\mid \mathcal{E}_1^c, \mathcal{E}_2^c, \mathcal{E}_3^c\right]&\geq 1-m\cdot Z_g(W^{\otimes \iota_m})\\
	&\geq 1-m\cdot \text{exp}\left({-\frac{\iota_m}{2}\cdot \mathsf{C}(W)+\Theta(\ln \left(\iota_m|\mathcal{X}|^\tau\right)}\right) =: 1-\zeta_{4,m}, \label{eq:dec4}
\end{align}
where $\mathsf{C}(W):= \min_{z\neq z'} \mathsf{C}(W_{Y|z},W_{Y|z'})>0$, with \begin{equation}\mathsf{C}(P_0,P_1):=-\min_{\lambda\in [0,1]}\ln \left(\sum_{x\in \mathcal{X}} P_0(x)^{1-\lambda}P_1(x)^\lambda\right)\end{equation} being the standard Chernoff distance \cite[Ch. 11]{cover_thomas} between distributions $P_0, P_1$ on the same alphabet. We mention that the  inequality in \eqref{eq:dec4} uses the upper bound on the Bhattacharya parameter via the conditional entropy in \cite[Prop. 4.8]{sasoglufnt} and \cite[Thm. 3.1]{arnw-t-it}.

Putting everything together, we obtain that
\begin{align}
	\lim_{m\to \infty}\Pr\left[\bigcap_{i=1}^4 \mathcal{E}_i^c\right] &\geq \lim_{m\to \infty} \prod_{i=1}^4 (1-\zeta_{i,m})\\
	&\geq 1-\lim_{m\to \infty} \sum_{i=1}^4 \zeta_{i,m} = 1,
\end{align}
implying that $\lim_{m\to \infty}\Pr\left[\textsc{Decode}(Y^{T_m})\neq S^m\right] = 0$, as required.
\end{proof}
As a direct corollary of Theorem \ref{thm:decodemain}, we obtain the following statement that shows the effectiveness of our algorithm in ``denoising" the nanopore channel $W_\text{nn}$, for sufficiently large sampling rates. 
\begin{corollary}
	\label{cor:decodemain}
	In the high sampling-rate regime \eqref{eq:sampling}, rates of up to $C_\tau^\text{\normalfont no-noise, no-loop}$ are achievable using the decoder in Algorithm \ref{alg:decode}.
\end{corollary}
\begin{proof}
	Let $P_e^{(m)}:=\Pr\left[\textsc{Decode}(Y^{T_m})\neq S^m\right]$. Observe that for any fixed no-self-loop de Bruijn Markov input process $P_{S|S^-}$, we have, using the decoder in Algorithm \ref{alg:decode}, that
	\begin{align}
		I(S^m;Y^{T_m})\geq H(S^m)-h_b\left(P_e^{(m)}\right)-P_e^{(m)}\cdot \log |\mathcal{X}|^\tau,
	\end{align}
due to Fano's inequality \cite[Thm. 2.10.1]{cover_thomas}. Since we have from Theorem \ref{thm:decodemain} that $\lim_{m\to \infty} P_e^{(m)} = 0$, the statement of the corollary follows.
\end{proof}
We remark that by choosing $\tau$ large enough, we can achieve rates that are arbitrarily close to the optimal rate of $1$, via Theorem \ref{thm:nonoisecap}. We mention also that our specific choices of parameters $\ell_m, h_m, \alpha_m, c_m$ can be changed suitably to other values, to ensure that $\lim_{m\to \infty} \Pr\left[\bigcap_{i=1}^4 \mathcal{E}_i^c\right] = 1$.
\section{Conclusion and Future Work}
\label{sec:conclusion}
In this paper, we continued the study of the noisy nanopore channel (NNC), introduced in \cite{mcbaininfo1}, and presented explicit achievable rates over the channel. In particular, we discussed a (tight) computable lower bound on the capacity of the \emph{noiseless} nanopore channel. We then discussed computable lower and upper bounds on the capacity of NNCs with general noise distributions. Future work calls for a sharpening of these bounds to be accurate in regimes of moderate/large memory length. Next, for an NNC with erasure noise, we showed that for large memory lengths, the capacity of the NNC can be made to approach $1$ arbitrarily closely. We then presented an explicit decoding algorithm for the regime of high sampling rates, which relies on a change-point detection procedure. We argue that using this decoder, one can achieve rates arbitrarily close to the noise-free capacity over such a channel.

An important direction for future research will be to tighten the non-asymptotic upper and lower bounds on the capacity in this paper. One can also try to extend our results on NNCs with large $\tau$-mer lengths from the setting of erasure noise to more general noise distributions. Another direction is to design explicit codes over general NNCs, for \emph{fixed} memory lengths and \emph{bounded} duplication noise, and analytically compute the rates they achieve.

\ifCLASSOPTIONcaptionsoff
  \newpage
\fi



%
\bibliographystyle{IEEEtran}
{\footnotesize
	\bibliography{references}}

\begin{thebibliography}{10}
\providecommand{\url}[1]{#1}
\csname url@samestyle\endcsname
\providecommand{\newblock}{\relax}
\providecommand{\bibinfo}[2]{#2}
\providecommand{\BIBentrySTDinterwordspacing}{\spaceskip=0pt\relax}
\providecommand{\BIBentryALTinterwordstretchfactor}{4}
\providecommand{\BIBentryALTinterwordspacing}{\spaceskip=\fontdimen2\font plus
\BIBentryALTinterwordstretchfactor\fontdimen3\font minus
  \fontdimen4\font\relax}
\providecommand{\BIBforeignlanguage}[2]{{%
\expandafter\ifx\csname l@#1\endcsname\relax
\typeout{** WARNING: IEEEtran.bst: No hyphenation pattern has been}%
\typeout{** loaded for the language `#1'. Using the pattern for}%
\typeout{** the default language instead.}%
\else
\language=\csname l@#1\endcsname
\fi
#2}}
\providecommand{\BIBdecl}{\relax}
\BIBdecl

\bibitem{dnachurch}
\BIBentryALTinterwordspacing
G.~M. Church, Y.~Gao, and S.~Kosuri, ``Next-generation digital information
  storage in {{DNA}},'' \emph{Science}, vol. 337, no. 6102, pp. 1628--1628,
  2012. [Online]. Available:
  \url{https://www.science.org/doi/abs/10.1126/science.1226355}
\BIBentrySTDinterwordspacing

\bibitem{dnagoldman}
N.~Goldman, P.~Bertone, S.~Chen, C.~Dessimoz, E.~M. LeProust, B.~Sipos, and
  E.~Birney, ``\BIBforeignlanguage{en}{Towards practical, high-capacity,
  low-maintenance information storage in synthesized {{DNA}}},''
  \emph{\BIBforeignlanguage{en}{Nature}}, vol. 494, no. 7435, pp. 77--80, Jan.
  2013.

\bibitem{dnagrass}
\BIBentryALTinterwordspacing
R.~N. Grass, R.~Heckel, M.~Puddu, D.~Paunescu, and W.~J. Stark, ``Robust
  chemical preservation of digital information on {{DNA}} in silica with
  error-correcting codes,'' \emph{Angewandte Chemie International Edition},
  vol.~54, no.~8, pp. 2552--2555, 2015. [Online]. Available:
  \url{https://onlinelibrary.wiley.com/doi/abs/10.1002/anie.201411378}
\BIBentrySTDinterwordspacing

\bibitem{dnaerlich}
\BIBentryALTinterwordspacing
Y.~Erlich and D.~Zielinski, ``{{DNA} Fountain} enables a robust and efficient
  storage architecture,'' \emph{Science}, vol. 355, no. 6328, pp. 950--954,
  2017. [Online]. Available:
  \url{https://www.science.org/doi/abs/10.1126/science.aaj2038}
\BIBentrySTDinterwordspacing

\bibitem{dnayazdi}
S.~M. H.~T. Yazdi, R.~Gabrys, and O.~Milenkovic,
  ``\BIBforeignlanguage{en}{Portable and error-free {DNA-based} data
  storage},'' \emph{\BIBforeignlanguage{en}{Sci. Rep.}}, vol.~7, no.~1, p.
  5011, Jul. 2017.

\bibitem{dnaorganick}
\BIBentryALTinterwordspacing
L.~Organick \emph{et~al.}, ``Random access in large-scale {{DNA}} data
  storage,'' \emph{Nature Biotechnology}, vol.~36, no.~3, pp. 242--248, Mar
  2018. [Online]. Available: \url{https://doi.org/10.1038/nbt.4079}
\BIBentrySTDinterwordspacing

\bibitem{shomorony}
\BIBentryALTinterwordspacing
I.~Shomorony and R.~Heckel, ``Information-theoretic foundations of {{DNA}} data
  storage,'' \emph{Foundations and Trends® in Communications and Information
  Theory}, vol.~19, no.~1, pp. 1--106, 2022. [Online]. Available:
  \url{http://dx.doi.org/10.1561/0100000117}
\BIBentrySTDinterwordspacing

\bibitem{nir_merhav}
N.~Weinberger and N.~Merhav, ``The {{DNA}} storage channel: Capacity and error
  probability bounds,'' \emph{IEEE Transactions on Information Theory},
  vol.~68, no.~9, pp. 5657--5700, 2022.

\bibitem{lenz}
A.~Lenz, P.~H. Siegel, A.~Wachter-Zeh, and E.~Yaakobi, ``The noisy drawing
  channel: Reliable data storage in {DNA} sequences,'' \emph{IEEE Transactions
  on Information Theory}, vol.~69, no.~5, pp. 2757--2778, 2023.

\bibitem{sequencingsurvey}
\BIBentryALTinterwordspacing
D.~Deamer, M.~Akeson, and D.~Branton, ``Three decades of nanopore sequencing,''
  \emph{Nature Biotechnology}, vol.~34, no.~5, pp. 518--524, May 2016.
  [Online]. Available: \url{https://doi.org/10.1038/nbt.3423}
\BIBentrySTDinterwordspacing

\bibitem{oxford}
\BIBentryALTinterwordspacing
{Oxford Nanopore Technologies}. [Online]. Available:
  \url{https://nanoporetech.com}
\BIBentrySTDinterwordspacing

\bibitem{yazdi-nanopore}
S.~K. Tabatabaei, B.~Pham, C.~Pan, J.~Liu, S.~Chandak, S.~A. Shorkey, A.~G.
  Hernandez, A.~Aksimentiev, M.~Chen, C.~M. Schroeder, and O.~Milenkovic,
  ``\BIBforeignlanguage{en}{Expanding the molecular alphabet of {DNA-based}
  data storage systems with neural network nanopore readout processing},''
  \emph{\BIBforeignlanguage{en}{Nano Lett.}}, vol.~22, no.~5, pp. 1905--1914,
  Mar. 2022.

\bibitem{chakra-nanopore}
R.~Chakraborty, M.~Xiong, N.~Athreya, S.~K. Tabatabaei, O.~Milenkovic, and
  J.-P. Leburton, ``\BIBforeignlanguage{en}{Solid-state {MoS$_{2}$} nanopore
  membranes for discriminating among the lengths of {RNA} tails on a
  double-stranded {DNA}: A new simulation-based differentiating algorithm},''
  \emph{\BIBforeignlanguage{en}{ACS Appl. Nano Mater.}}, vol.~6, no.~6, pp.
  4651--4660, Mar. 2023.

\bibitem{mao}
W.~Mao, S.~N. Diggavi, and S.~Kannan, ``Models and information-theoretic bounds
  for nanopore sequencing,'' \emph{IEEE Transactions on Information Theory},
  vol.~64, no.~4, pp. 3216--3236, 2018.

\bibitem{chandaknanopore}
R.~Hulett, S.~Chandak, and M.~Wootters, ``On coding for an abstracted nanopore
  channel for {DNA} storage,'' in \emph{2021 IEEE International Symposium on
  Information Theory (ISIT)}, 2021, pp. 2465--2470.

\bibitem{hamoumnanopore}
B.~Hamoum, E.~Dupraz, L.~Conde-Canencia, and D.~Lavenier, ``Channel model with
  memory for {DNA} data storage with nanopore sequencing,'' in \emph{2021 11th
  International Symposium on Topics in Coding (ISTC)}, 2021, pp. 1--5.

\bibitem{mcbaininfo1}
B.~McBain, E.~Viterbo, and J.~Saunderson, ``Information rates of the noisy
  nanopore channel,'' \emph{IEEE Transactions on Information Theory}, vol.~70,
  no.~8, pp. 5640--5652, 2024.

\bibitem{mcbainsurvey}
B.~McBain and E.~Viterbo, ``An information-theoretic approach to nanopore
  sequencing for {DNA} storage,'' \emph{IEEE BITS the Information Theory
  Magazine}, vol.~3, no.~3, pp. 95--108, 2023.

\bibitem{scrappie}
\BIBentryALTinterwordspacing
Scrappie technology demonstrator. [Online]. Available:
  \url{https://github.com/nanoporetech/scrappie}
\BIBentrySTDinterwordspacing

\bibitem{mcbainmodel}
B.~McBain, E.~Viterbo, and J.~Saunderson, ``Finite-state semi-{Markov} channels
  for nanopore sequencing,'' in \emph{2022 IEEE International Symposium on
  Information Theory (ISIT)}, 2022, pp. 216--221.

\bibitem{mcbaininfo2}
B.~McBain, J.~Saunderson, and E.~Viterbo, ``On noisy duplication channels with
  {Markov} sources,'' in \emph{2024 IEEE International Symposium on Information
  Theory (ISIT)}, 2024, pp. 3438--3443.

\bibitem{Dob67}
R.~L. Dobrushin, ``Shannon's theorems for channels with synchronization
  errors,'' \emph{Problemy Peredachi Informatsii}, vol.~3, no.~4, pp. 11--26,
  1967.

\bibitem{diggavidel}
S.~Diggavi and M.~Grossglauser, ``On information transmission over a finite
  buffer channel,'' \emph{IEEE Transactions on Information Theory}, vol.~52,
  no.~3, pp. 1226--1237, 2006.

\bibitem{diggaviupper}
S.~Diggavi, M.~Mitzenmacher, and H.~D. Pfister, ``Capacity upper bounds for the
  deletion channel,'' in \emph{2007 IEEE International Symposium on Information
  Theory}, 2007, pp. 1716--1720.

\bibitem{drinea1}
E.~Drinea and M.~Mitzenmacher, ``On lower bounds for the capacity of deletion
  channels,'' \emph{IEEE Transactions on Information Theory}, vol.~52, no.~10,
  pp. 4648--4657, 2006.

\bibitem{drinea3}
M.~Mitzenmacher and E.~Drinea, ``A simple lower bound for the capacity of the
  deletion channel,'' \emph{IEEE Transactions on Information Theory}, vol.~52,
  no.~10, pp. 4657--4660, 2006.

\bibitem{drinea2}
E.~Drinea and M.~Mitzenmacher, ``Improved lower bounds for the capacity of
  i.i.d. deletion and duplication channels,'' \emph{IEEE Transactions on
  Information Theory}, vol.~53, no.~8, pp. 2693--2714, 2007.

\bibitem{kirsch}
A.~Kirsch and E.~Drinea, ``Directly lower bounding the information capacity for
  channels with i.i.d. deletions and duplications,'' \emph{IEEE Transactions on
  Information Theory}, vol.~56, no.~1, pp. 86--102, 2010.

\bibitem{fertonani}
D.~Fertonani and T.~M. Duman, ``Novel bounds on the capacity of the binary
  deletion channel,'' \emph{IEEE Transactions on Information Theory}, vol.~56,
  no.~6, pp. 2753--2765, 2010.

\bibitem{aravind}
A.~R. Iyengar, P.~H. Siegel, and J.~K. Wolf, ``On the capacity of channels with
  timing synchronization errors,'' \emph{IEEE Transactions on Information
  Theory}, vol.~62, no.~2, pp. 793--810, 2016.

\bibitem{mercier}
H.~Mercier, V.~Tarokh, and F.~Labeau, ``Bounds on the capacity of discrete
  memoryless channels corrupted by synchronization and substitution errors,''
  \emph{IEEE Transactions on Information Theory}, vol.~58, no.~7, pp.
  4306--4330, 2012.

\bibitem{rahmati1}
M.~Rahmati and T.~M. Duman, ``Achievable rates for noisy channels with
  synchronization errors,'' \emph{IEEE Transactions on Communications},
  vol.~62, no.~11, pp. 3854--3863, 2014.

\bibitem{coding-1}
\BIBentryALTinterwordspacing
B.~Haeupler and A.~Shahrasbi, ``Synchronization strings: Codes for insertions
  and deletions approaching the singleton bound,'' \emph{J. ACM}, vol.~68,
  no.~5, Sep. 2021. [Online]. Available: \url{https://doi.org/10.1145/3468265}
\BIBentrySTDinterwordspacing

\bibitem{coding-2}
V.~Guruswami and R.~Li, ``Polynomial time decodable codes for the binary
  deletion channel,'' \emph{IEEE Transactions on Information Theory}, vol.~65,
  no.~4, pp. 2171--2178, 2019.

\bibitem{coding-3}
\BIBentryALTinterwordspacing
F.~Pernice, R.~Li, and M.~Wootters, ``Efficient capacity-achieving codes for
  general repeat channels,'' in \emph{2022 IEEE International Symposium on
  Information Theory (ISIT)}.\hskip 1em plus 0.5em minus 0.4em\relax IEEE
  Press, 2022, p. 3097–3102. [Online]. Available:
  \url{https://doi.org/10.1109/ISIT50566.2022.9834386}
\BIBentrySTDinterwordspacing

\bibitem{coding-4}
R.~Con and A.~Shpilka, ``Improved constructions of coding schemes for the
  binary deletion channel and the {Poisson} repeat channel,'' \emph{IEEE
  Transactions on Information Theory}, vol.~68, no.~5, pp. 2920--2940, 2022.

\bibitem{sloanedel}
\BIBentryALTinterwordspacing
N.~J.~A. Sloane, \emph{On single-deletion-correcting codes}.\hskip 1em plus
  0.5em minus 0.4em\relax Berlin, New York: De Gruyter, 2002, pp. 273--292.
  [Online]. Available: \url{https://doi.org/10.1515/9783110198119.273}
\BIBentrySTDinterwordspacing

\bibitem{laszlo}
\BIBentryALTinterwordspacing
A.~H. Laszlo \emph{et~al.}, ``Decoding long nanopore sequencing reads of
  natural {DNA},'' \emph{Nature Biotechnology}, vol.~32, no.~8, pp. 829--833,
  Aug 2014. [Online]. Available: \url{https://doi.org/10.1038/nbt.2950}
\BIBentrySTDinterwordspacing

\bibitem{han-verdu}
S.~Verdu and T.~S. Han, ``A general formula for channel capacity,'' \emph{IEEE
  Transactions on Information Theory}, vol.~40, no.~4, pp. 1147--1157, 1994.

\bibitem{thinning}
P.~Harremoes, O.~Johnson, and I.~Kontoyiannis, ``Thinning and the law of small
  numbers,'' in \emph{2007 IEEE International Symposium on Information Theory},
  2007, pp. 1491--1495.

\bibitem{tao-sum}
T.~Tao, ``Sumset and inverse sumset theory for {Shannon} entropy,''
  \emph{Combinatorics, Probability and Computing}, vol.~19, no.~4, p.
  603–639, 2010.

\bibitem{madiman1}
M.~Madiman, ``On the entropy of sums,'' in \emph{2008 IEEE Information Theory
  Workshop}, 2008, pp. 303--307.

\bibitem{entropycheraghchi}
M.~Cheraghchi, ``Expressions for the entropy of binomial-type distributions,''
  in \emph{2018 IEEE International Symposium on Information Theory (ISIT)},
  2018, pp. 2520--2524.

\bibitem{mitzenmacher_datarecovery}
M.~Mitzenmacher, ``On the theory and practice of data recovery with multiple
  versions,'' in \emph{2006 IEEE International Symposium on Information
  Theory}, 2006, pp. 982--986.

\bibitem{landhuber}
\BIBentryALTinterwordspacing
I.~Land and J.~Huber, ``Information combining,'' \emph{Foundations and Trends®
  in Communications and Information Theory}, vol.~3, no.~3, pp. 227--330, 2006.
  [Online]. Available: \url{http://dx.doi.org/10.1561/0100000013}
\BIBentrySTDinterwordspacing

\bibitem{arnw-t-it}
V.~A. Rameshwar and N.~Weinberger, ``Information rates over multi-view
  channels,'' \emph{IEEE Transactions on Information Theory}, vol.~71, no.~2,
  pp. 847--861, 2025.

\bibitem{sasoglufnt}
\BIBentryALTinterwordspacing
E.~Şaşoğlu, ``Polarization and polar codes,'' \emph{Foundations and Trends®
  in Communications and Information Theory}, vol.~8, no.~4, pp. 259--381, 2012.
  [Online]. Available: \url{http://dx.doi.org/10.1561/0100000041}
\BIBentrySTDinterwordspacing

\bibitem{mitzupfal}
M.~Mitzenmacher and E.~Upfal, \emph{Probability and Computing: Randomized
  Algorithms and Probabilistic Analysis}.\hskip 1em plus 0.5em minus
  0.4em\relax Cambridge University Press, 2005.

\bibitem{cover_thomas}
T.~M. Cover and J.~A. Thomas, \emph{Elements of Information Theory},
  2nd~ed.\hskip 1em plus 0.5em minus 0.4em\relax Wiley-India, 2010.

\bibitem{Roth}
\BIBentryALTinterwordspacing
B.~H. Marcus, R.~M. Roth, and P.~H. Siegel, ``An introduction to coding for
  constrained systems,'' lecture notes. [Online]. Available:
  \url{https://ronny.cswp.cs.technion.ac.il/wp-content/uploads/sites/54/2016/05/chapters1-9.pdf}
\BIBentrySTDinterwordspacing

\bibitem{veeravalli}
V.~V. Veeravalli and T.~Banerjee, \emph{Quickest change detection}.\hskip 1em
  plus 0.5em minus 0.4em\relax Elsevier, 2014, vol.~3, pp. 209--255.

\bibitem{shiryaev1}
\BIBentryALTinterwordspacing
A.~N. Shiryaev, ``On optimum methods in quickest detection problems,''
  \emph{Theory of Probability \& Its Applications}, vol.~8, no.~1, pp. 22--46,
  1963. [Online]. Available: \url{https://doi.org/10.1137/1108002}
\BIBentrySTDinterwordspacing

\bibitem{tartakovsky-veer}
\BIBentryALTinterwordspacing
A.~G. Tartakovsky and V.~V. Veeravalli, ``General asymptotic {Bayesian} theory
  of quickest change detection,'' \emph{Theory of Probability \& Its
  Applications}, vol.~49, no.~3, pp. 458--497, 2005. [Online]. Available:
  \url{https://doi.org/10.1137/S0040585X97981202}
\BIBentrySTDinterwordspacing

\end{thebibliography}
\appendices
\section{Proof of Lemma \ref{lem:eig}}
\label{app:eig}
\begin{proof}[Proof of Lemma \ref{lem:eig}]
Without loss of generality, assume that $\mathcal{X} = \{0,1,\ldots,q-1\}$, for some positive integer $q$. Consider the adjacency matrix $A_{G_\tau^\text{no-loop}}$ for some fixed $\tau\geq 2$. We first reorder the rows and columns of $A_{G_\tau^\text{no-loop}}$ so that they are indexed by states $s\in \mathcal{X}^\tau$ in the standard lexicographic order on strings in $\mathcal{X}^\tau$, i.e., if $\mathbf{z} = (z_1,\ldots,z_\tau)$ and $\mathbf{z}^{\prime} = (z_1^{\prime},\ldots,z_\tau^{\prime})$ are two states, then, $\mathbf{z}$ occurs before $\mathbf{z}^{\prime}$ iff for some $i\geq 1$, we have $z_j = z_j^{\prime}$ for all $j<i$, and $z_i < z_i^{\prime}$. 

Let
\begin{equation}
A_{G_\tau^\text{no-loop}} = 
\begin{bmatrix}
	A_1 & B_{1,1} & B_{1,2} & \ldots & B_{1,|\mathcal{X}|-1}\\
	B_{2,1} & A_2 & B_{2,2} &\ldots & B_{2,|\mathcal{X}|-1}\\
	\vdots &\vdots &\vdots &\ddots &\vdots\\
	B_{|\mathcal{X}|,1} & B_{|\mathcal{X}|,2} & B_{|\mathcal{X}|,3} & \ldots & A_{|\mathcal{X}|}
\end{bmatrix},
\end{equation}
where each $A_i$, $i\in [q]$ and each $B_{i,j}$, $i\in [q], j\in [q-1]$, is a matrix of order $q^{\tau-1}\times q^{\tau-1}$. By the structure of de Bruijn Markov processes (without self-loops), it can be checked that
\begin{equation}
A_{G_{\tau-1}^\text{no-loop}}  = \sum_{i=1}^{q} A_i.
\end{equation}
To see why, observe that via our ordering of states, each $A_i$, $i\in [q]$, is such exactly $q^{\tau-2}$ of its rows have non-zero entries; these are precisely those rows $\mathbf{z}\in \mathcal{X}^{\tau}$ of $A_{G_\tau^\text{no-loop}}$ lying in $A_i$ (i.e., with $z_1 = i-1$) such that $z_2 = i-1$.
Now, let us define
\begin{equation}
\overline{A}:= 
\begin{bmatrix}
	A_1 & \mathbf{0} & \mathbf{0} & \ldots & \mathbf{0}\\
	\mathbf{0}& A_2 & \mathbf{0} &\ldots & \mathbf{0}\\
	\vdots &\vdots &\vdots &\ddots &\vdots\\
	\mathbf{0}& \mathbf{0} & \mathbf{0} & \ldots & A_{|\mathcal{X}|}
\end{bmatrix},
\end{equation}
where $\mathbf{0}$ denotes the all-zeros matrix of order $|\mathcal{X}|^{\tau-1}\times |\mathcal{X}|^{\tau-1}$. Let $\overline{G}_\tau$ denote the directed graph whose adjacency matrix is $\overline{A}$.

From \cite[Problem 3.26]{Roth} and \cite[Prop. 3.12]{Roth}, we obtain that $\lambda(A_{G_\tau^\text{no-loop}})>\lambda(\overline{A})$. We now claim that $\lambda(\overline{A})\geq A_{G_{\tau-1}^\text{no-loop}}$. In what follows, we prove this claim. Let $P_{t}^\star$ be the transition kernel corresponding to a max-entropic Markov chain supported on the graph $G_{t}^\text{no-loop}$, for any $t\geq 1$, and let $H(P_{\tau-1}^\star)$ denote its entropy rate. Also, let $\overline{P}_\tau^\star$ denote the max-entropic Markov chain supported on $\overline{G}_\tau$. Further, for each $i\in \mathcal{X}$, let $\mathcal{S}_i$ denote the collection of states $\mathbf{z}\in \mathcal{X}^\tau$ with $z_1 = z_2= i-1$; recall that these are precisely the rows of $A_{G_\tau^\text{no-loop}}$ lying in $A_i$ that have at least one non-zero entry. 

The following sequence of inequalities then holds:
\begin{align}
	\log_{|\mathcal{X}|} (\lambda(A_{G_{\tau-1}^\text{no-loop}}))&=H(P_{\tau-1}^\star)\\
	&=\sum_{i=1}^{q} H(S\mid S^-\in \mathcal{S}_i)\Pr[S^-\in \mathcal{S}_i]\\
	&\leq \max_{i\in [q]} H(S\mid S^-\in \mathcal{S}_i)\\
	&\leq H(\overline{P}_{\tau}^\star)= \log_{|\mathcal{X}|} (\lambda(\overline{A})),
\end{align}
implying that $\lambda(\overline{A})\geq \lambda(A_{G_{\tau-1}^\text{no-loop}})$. Here, the second inequality holds via the structure of Markov chains supported on $\overline{G}_\tau$ (see also \cite[Thm. 3.1]{Roth}). Finally, using the fact that $\lambda(\overline{A})<\lambda(A_{G_\tau^\text{no-loop}})$, we complete the proof of the lemma.
\end{proof}

%

%





\end{document}